\documentclass[10pt,journal]{IEEEtran}
\usepackage{amsmath,amsfonts,amssymb,amsthm}
\usepackage[linesnumbered,ruled,vlined]{algorithm2e}
\usepackage{array}
\usepackage[caption=false,font=normalsize,labelfont=sf,textfont=sf]{subfig}
\usepackage{textcomp,epsfig,url}
\usepackage{stfloats}
\usepackage{url}
\usepackage{verbatim}
\usepackage{graphicx}
\usepackage[usenames, dvipsnames]{color}  
\usepackage{multirow}
\usepackage{cite}
\usepackage{color}
\usepackage{adjustbox}
\usepackage{tabularx}
\newtheorem{theorem}{Theorem}

\newtheorem{proposition}[theorem]{Proposition}

\newtheorem{definition}[theorem]{Definition}

\SetKwInOut{Parameter}{Parameter}

\def\btheta{\boldsymbol{\theta}}

\def\s{\boldsymbol{s}}
\def\b{\boldsymbol{b}}
\def\q{\boldsymbol{q}}
\def\r{\boldsymbol{r}}
\def\d{\boldsymbol{d}}

\def\x{\boldsymbol{x}}
\def\e{\boldsymbol{e}}
\def\v{\boldsymbol{v}}
\def\y{\boldsymbol{y}}
\def\z{\boldsymbol{z}}

\def\s{\boldsymbol{s}}
\def\t{\boldsymbol{t}}
\def\u{\boldsymbol{u}}
\def\v{\boldsymbol{v}}

\def\A{\mathbf{A}}
\def\W{\mathbf{W}}
\def\P{\mathbf{P}}
\def\B{\mathbf{B}}
\def\H{\mathbf{H}}
\def\M{\mathbf{M}}

\def\Q{\mathbf{Q}}
\def\I{\mathbf{I}}

\def\Y{\mathbf{Y}}
\def\Z{\mathbf{Z}}
\def\D{\mathbf{D}}
\def\V{\mathbf{V}}
\def\F{\mathbf{F}}

\def\K{\mathbf{K}}
\def\S{\mathbf{S}}
\def\U{\mathbf{U}}
\def\N{\mathcal{N}}
\def\C{\mathbf{C}}
\def\R{\mathcal{R}}

\def\Feat{\mathcal{F}}
\def\BLambda{\boldsymbol{\Lambda}}
\def\Bzeta{\boldsymbol{\zeta}}
\def\prox{\mathrm{prox}}

\def\argmin{\mathop{\mathrm{argmin}}}

\def\bxi{\boldsymbol{\xi}}
\renewcommand{\Re}{\mathbb{R}}
\newcommand{\inner}[2]{\left\langle#1,#2\right\rangle}
\newcommand{\norm}[1]{\left\lVert#1\right\rVert}
  
\newcolumntype{Y}{>{\centering\arraybackslash}X}

\begin{document}
\markboth{Submitted to IEEE Transactions on Image Processing}{}

\title{Plug-and-Play Regularization using Linear Solvers}

\author{Pravin Nair,~\IEEEmembership{Student~Member,~IEEE} and Kunal N. Chaudhury,~\IEEEmembership{Senior~Member,~IEEE}

\thanks{The work of K.~N.~Chaudhury was supported by  grant CRG/2020/000527  from SERB, Government of India.}
\thanks{P.~Nair and K.~N.~Chaudhury are with the Department of Electrical Engineering, Indian Institute of Science, Bengaluru 560012, India. Email: pravinn@iisc.ac.in, kunal@iisc.ac.in.}
}

\maketitle

\begin{abstract}
There has been tremendous research on the design of image regularizers over the years, from simple Tikhonov and Laplacian to sophisticated sparsity and CNN-based regularizers. Coupled with a model-based loss function, these are typically used for image reconstruction within an optimization framework. The technical challenge is to develop a regularizer that can accurately model realistic images and be optimized efficiently along with the loss function. Motivated by the recent plug-and-play paradigm for image regularization, we construct a quadratic regularizer whose reconstruction capability is competitive with state-of-the-art regularizers. The novelty of the regularizer is that, unlike classical  regularizers, the quadratic objective function is derived from the observed data. Since the regularizer is quadratic, we can reduce the optimization to solving a linear system for applications such as superresolution, deblurring, inpainting, etc. In particular, we show that using iterative Krylov solvers, we can converge to the solution in few iterations, where each iteration requires an application of the forward operator and a linear denoiser. The surprising finding is that we can get close to deep learning methods in terms of reconstruction quality. To the best of our knowledge, the possibility of achieving near state-of-the-art performance using a linear solver is novel.
\end{abstract}
 

\section{Introduction}

Several image reconstruction problems such as  deblurring, superresolution, inpainting, compressed sensing, demosaicking, etc. are modeled as linear inverse problems where we wish to recover an unknown image $\btheta \in \Re^n$ from noisy linear measurements $\y \in \Re^m$ of the form
\begin{equation}
\label{model}
\y = \F \btheta + \bxi,             
\end{equation}
where $\F \in \Re^{m \times n}$ is the forward  transform (imaging model) and $\bxi \in \Re^m$ is white Gaussian noise \cite{ribes2008linear,scherzer2009variational}. Typically, the reconstruction is performed by solving the optimization problem
\begin{equation}
\label{opt}
\min_{\x \in \Re^n} \ f(\x) +  g(\x),
\end{equation}
where
\begin{equation}
\label{loss}
f(\x)=\frac{1}{2} \| \F\x - \y\|^2
\end{equation}
is the quadratic loss and $g: \Re^n \to \Re \cup \{\infty\}$ is a regularizer \cite{ribes2008linear}. Needless to say, the reconstruction depends crucially on the choice of the regularizer. Starting with simple Tikhonov and Laplacian regularizers \cite{chan2005image,aubert2006mathematical}, we have progressed from wavelet and total-variation \cite{chambolle1998nonlinear,rudin1992nonlinear} to dictionary and CNN regularizers \cite{scherzer2009variational,krishnan2009fast,zoran2011learning,dong2012nonlocally}. More recently, it has been shown that powerful denoisers such as  NLM \cite{buades2005non}, BM3D \cite{dabov2007image} and DnCNN \cite{zhang2017beyond} can be used for regularization purpose. Two prominent techniques in this regard are Regularization-by-Denoising (RED) \cite{romano2017little,mataev2019deepred}  and Plug-and-Play (PnP)  \cite{sreehari2016plug,CWE2017}. 

\subsection{PnP regularization}

The present work is motivated by PnP, where we start with a proximal algorithm for optimizing \eqref{opt} such as ISTA or ADMM and  
replace the proximal operator associated with $g$ (which effectively acts as denoiser) by a more powerful denoiser, albeit in an ad-hoc fashion. 
 For example, in ISTA \cite{beck2009fast}, starting with an initialization $\x_0 \in \Re^n$, the updates are performed as
\begin{figure}
\centering
\subfloat[Observation.]{\includegraphics[width=0.48\linewidth]{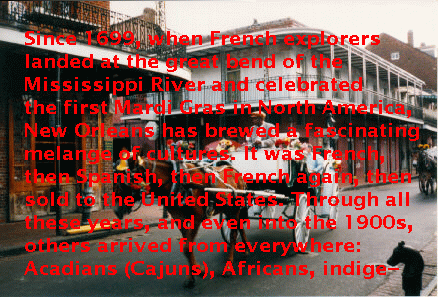}} \hspace{0.1mm}
\subfloat[Output.]{\includegraphics[width=0.48\linewidth]{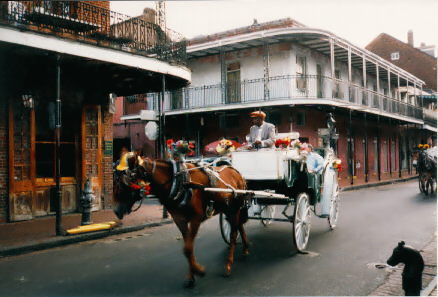}}
\caption{Text removal by Algo.~\ref{propalgo}, which requires solving just one linear system.}

\label{inpaintreal}
\end{figure}

\begin{equation}
\label{ista}
\x_{k+1} = \mathrm{prox}_{g} \big(\x_k- \rho^{-1} \nabla \! f(\x_k)\big), 
\end{equation}
where $\mathrm{prox}_{g}$ is the proximal operator of $g$ and $\rho > 0$ is the step size. In PnP regularization, instead of going through $g$ and its proximal operator, $\mathrm{prox}_{g}$ in \eqref{ista} is directly replaced by a denoiser, i.e., the update is performed as
\begin{equation}
\label{pnpista}
\x_{k+1} = \mathfrak{D} \big(\x_k- \rho^{-1} \nabla \!f(\x_k)\big),
\end{equation}
where $\mathfrak{D}: \Re^n \to \Re^n$ is a denoising operator. The same idea can be applied to other iterative algorithms such as FISTA, Douglas Rachford Splitting, ADMM, Approximate Massage Passing and Chambolle-Pock \cite{sreehari2016plug,gavaskar2021plug,nair2021fixed,zhang2017learning, Ryu2019_PnP_trained_conv,borgerding2017amp,Sun2019_PnP_SGD,sun2021scalable,ono2017primal}. For example, in PnP-ADMM \cite{sreehari2016plug}, starting with $\v_0, \z_0 \in \Re^n$, the updates are performed as
\begin{align}
\label{pnpadmm}
& \x_{k+1} = \underset{\x}{\mathrm{argmin}} \ f(\x) + \frac{1}{2 \rho}\| \x - (\v_k - \z_k) \|^2, \\
& \v_{k+1} = \mathfrak{D} (\x_{k+1} + \z_k), \nonumber \\ 
& \z_{k+1} = \z_k + (\x_{k+1} - \v_{k+1}), \nonumber
\end{align}
where $\rho > 0$ is a penalty parameter. PnP has been shown to yield impressive results for many imaging problems \cite{sreehari2016plug,zhang2017learning,CWE2017} and for graph signal processing \cite{yazaki2019interpolation,gavaskar2022regularization,nagahama2021graph}.

Since PnP algorithms demonstrate impressive regularization capabilities, convergence analysis of PnP has garnered wide interest. For example, iterate convergence has been established for linear inverse problems in \cite{Dong2018_DNN_prior,tirer2018image,gavaskar2020plug}. In these works,  either the denoiser or the loss function is constrained to satisfy some conditions. In particular, the denoiser is assumed to satisfy a descent condition in \cite{Dong2018_DNN_prior}, a boundedness condition in \cite{tirer2018image}, a linearity condition in \cite{gavaskar2020plug}, an averaged property in \cite{sun2019online,nair2021fixed,gavaskar2021plug,sun2021scalable} and demicontractivity in \cite{cohen2021regularization}. It was shown in \cite{Ryu2019_PnP_trained_conv} that PnP convergence is guaranteed for ISTA and ADMM for a specially trained CNN denoiser, provided the data-fidelity is strongly convex. Apart from \cite{Ryu2019_PnP_trained_conv}, PnP convergence has been established for CNN denoisers \cite{Meinhardt2017_learning_prox_op}, \cite{buzzard2018plug}, generative denoisers \cite{jagatap2019advance} and GAN projectors \cite{raj2019gan}. 

A natural question is whether the reconstruction obtained using the PnP method is optimal in some sense; in particular,  is the reconstruction a solution to some regularization problem of the form \eqref{opt}?   Till date, this question has only been resolved for linear denoisers. This includes kernel filters  such as Yaroslavsky \cite{yaroslavsky1985digital}, Lee \cite{lee1983digital}, bilateral \cite{tomasi1998bilateral}, nonlocal means (NLM) \cite{buades2005non} and LARK \cite{takeda2007kernel}, and non-kernel filters such as GMM \cite{teodoro2019image} and GLIDE \cite{talebi2013global}. More specifically, it is known that  if the linear denoiser is symmetric, then the PnP-ISTA and PnP-ADMM updates in \eqref{pnpista} and \eqref{pnpadmm} amounts to minimizing an objective of the form in \eqref{opt}, e.g., see \cite{sreehari2016plug,teodoro2019image}. 
It was later observed in  \cite{nair2021fixed,gavaskar2021plug} that one can associate a regularizer with a wider class of non-symmetric denoisers, including kernel denoisers such as NLM \cite{buades2005non}. More precisely, it was shown that if the linear operator $\mathfrak{D}$ is diagonalizable with eigenvalues in $[0, 1]$, then we can represent $\mathfrak{D}$ as the proximal operator of a convex regularizer $\Phi$. Two important points in this regard are (i) the effective domain of $\Phi$  is a subspace of $\Re^n$  and (ii) the norm in the definition of the proximal operator is a weighted $\ell_2$ norm and not the standard $\ell_2$ norm (the exact description is provided in Section \ref{qr}). The implication of this result is that performing the PnP updates in  \eqref{pnpista} and \eqref{pnpadmm} with a diagonalizable linear denoiser with eigenvalues in $[0,1]$ amounts to minimizing an objective function of the form $f+\Phi$. As will be made precise later, the regularizer $\Phi$ is derived from a linear denoiser which in turn is computed from some surrogate of the ground-truth image (which is derived from measurements).

\begin{figure}
\includegraphics[width=1.\linewidth]{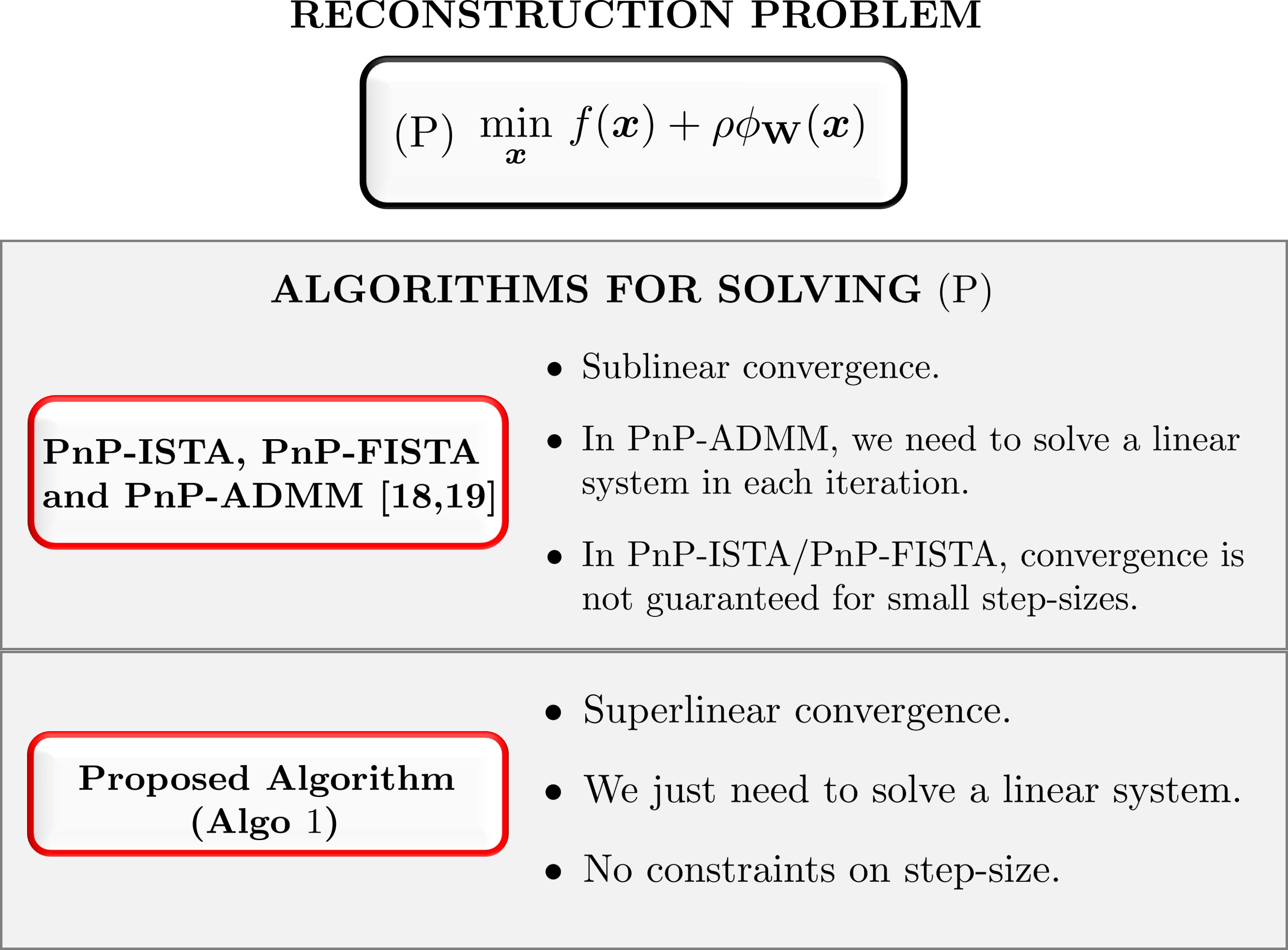}
\caption{Relation of the proposed algorithm with the PnP algorithms in   \cite{nair2021fixed,gavaskar2021plug}  for a linear denoiser $\W$. The reconstruction returned by such algorithms is a minimizer of a convex objective function $f+\Phi_{\W}$. Instead of iteratively computing the minimizer (as done in PnP algorithms), we propose to solve the first-order optimality condition for the problem, which can be reduced to solving a linear system if $f$ is quadratic. See Table \ref{compTIP21andourpaper} for more details.}
\label{contribution}
\end{figure}

\subsection{Motivation and contribution}

The present work is best motivated using an analogy with least-squares regression. Suppose that we want to minimize  ${\lVert \H \x - \b \rVert}^2$ with respect to $\x$. Among other things, we can compute the minimizer using either an iterative algorithm such as the gradient descent method or via the solution of the normal equation $\H^\top \!\H=\H^\top \! \b$ \cite{trefethen1997numerical}. Recall that the normal equation is the first-order optimality condition for the optimization problem at hand. In particular, if $\H$ has full column rank, then we can solve the normal equation using conjugate gradient, which has superior convergence than the sublinearly-convergent gradient descent \cite{shewchuk1994introduction}. 
Similarly, our main idea is that instead of the  iterative minimization in \eqref{pnpista} and \eqref{pnpadmm}, which in effect  minimizes $f+\Phi$, we can directly work with the first-order optimality condition and come up with an efficient algorithm (see Fig.~\ref{contribution}). The technical novelty in this regard are: 
\begin{itemize}
\item We prove that the minimization of $f+\Phi$ is well-posed, i.e., $f+\Phi$ is guaranteed to have a minimizer $\x^\ast \in \Re^n$ (Proposition \ref{minimizers}). Moreover, we prove that $\x^\ast$ is unique  for some imaging problems (Proposition \ref{uniqueness}).
\item We show that $\x^\ast$ can be computed by solving $\A\x=\b$. In particular, we prove that $\A$ is nonsingular for some reconstruction problems.
 
\item On the computational side, we show how $\x^\ast$ can be accurately computed in just few iterations using superlinearly convergent solvers such as GMRES \cite{saad1986gmres}, LGMRES \cite{baker2005technique}, GCROT \cite{hicken2010simplified} and Broyden \cite{broyden1965class,more1976global}, where each iteration requires us to apply $\F$, $\F^\top$ and a kernel denoiser.  Stated differently, we can obtain the same reconstruction as PnP-ISTA and PnP-ADMM, but unlike  PnP-ISTA and PnP-ADMM which are known to converge sublinearly \cite{beck2009fast,he20121}, we can achieve superlinear convergence (see Fig.~\ref{contribution} for a comparison of our and existing PnP solvers). The speedup is apparent in practice as well. Moreover, thanks to the superior convergence rate, the total inference time of our solvers (number of iterations  $\times$ time per iteration) are  comparable with PnP algorithms with CNN-based denoisers, though CNN-based denoisers are faster than kernel denoisers.

\end{itemize}
A sample inpainting result obtained using our algorithm is shown in Fig.~\ref{inpaintreal}, where the text is almost completely removed from the image. Surprisingly, for superresolution, deblurring, and inpainting, we are able to get close to state-of-the-art methods in terms of reconstruction quality. This demonstrates that as a natural image prior, $\Phi$ is more effective than classical wavelet and total-variation priors. Understanding this in greater depth is left to future work.
 
\subsection{Organization}
The paper is organized as follows. We present some background material about  kernel denoisers in Section \ref{kd} and existing results on the optimality of PnP algorithms for kernel denoisers in Section \ref{PnPkd}. We then explain the proposed work in Section \ref{kr} in three different subsections. In Section \ref{qr}, we provide a different formulation of the regularizer associated with the PnP algorithms in  \cite{nair2021fixed,gavaskar2021plug}. In Section \ref{reduction}, we show how the task of minimizing $f+\Phi$ can be reduced to solving a linear system. In Section III-C, we present our algorithm and perform  ablation studies on it; in particular, we experiment with different numerical solvers and demonstrate their effectiveness over existing PnP solvers. In Section \ref{exp}, we present results for superresolution, deblurring, and inpainting, and compare our method with top-performing methods. We conclude with a discussion of our findings in Section \ref{conc}. Proofs of technical results are deferred to Section \ref{Appendix}. 

\section{Background}
\label{br}

\subsection{Kernel denoiser}
\label{kd}
 
A kernel filter can abstractly be viewed as a linear transform $\W \in \Re^{n \times n}$ that takes a noisy image $\x \in \Re^n$ and returns the denoised output  $\W \x$, where $\x$ is some vectorized form of the image and $n$ is the number of pixels. We now explain this more precisely which will require some notations.

Let us denote the input image as $\{X_{\s}: \s \in \Omega\}$, where $\Omega  \subset \mathbb{Z}^2$ is the support of the image and $X_{\s}$ is the intensity value at pixel location $\s$. The essential point is that $\W$ is defined using a symmetric positive definite kernel $\phi : \Feat \times \Feat \to \Re$, where the feature space $\Feat$ is stipulated using a guide image $\{\u_{\s}: \s \in \Omega\}$. This is usually the input image for image denoising \cite{buades2005non}. However, in PnP regularization, the guide $\u$ is a surrogate of the ground-truth image $\btheta$ which is computed from the measurement \cite{sreehari2016plug,nair2021fixed, gavaskar2021plug}. The feature space in NLM is $\Feat= \Re^{2+p}$, where $p$ is the size of a pre-fixed square patch $P$ centered at the origin. In particular, the feature vector $\Bzeta_{\s} \in \Feat$ at pixel $\s \in \mathbb{Z}^2$ is obtained by concatenating the pixel position $\s$ and the intensity values $\{\u_{\s + \boldsymbol{\tau}}: \boldsymbol{\tau} \in P\} \in \Re^p$ from a patch around $\s$ in the guide image. The kernel in NLM is typically Gaussian and is defined as $\phi(\q,\r)=G(\q-\r)$, where  $G$ is a multivariate Gaussian   on $\Feat $ \cite{milanfar2013symmetrizing,Morel2014}.
The NLM operator $\mathfrak{D}$ which takes input $X$ and outputs image $\mathfrak{D}(X)$ is defined as
\begin{equation}
\label{neighbour}
\forall \ \s \in \Omega: \quad {(\mathfrak{D}(X))}_{\s} = \frac{\sum_{\t \in \Omega} \phi({\Bzeta}_{\s}, {\Bzeta}_{\t}) X_{\t}}{\sum_{\t \in \Omega} \phi({\Bzeta}_{\s}, {\Bzeta}_{\t})}.
\end{equation}
In other words, the output at a given pixel is obtained using a weighted average of its neighboring pixels, where the weights are derived from the similarity between features  $\Bzeta_{\s}$ and $\Bzeta_{\t}$ measured using kernel $\phi$.

It is clear from \eqref{neighbour} that $\mathfrak{D}$ is a linear operator. To obtain the matrix representation of \eqref{neighbour}, we 
need to linearly index the pixel space $\Omega$. Let $\sigma: [1, n] \to \Omega$ be such that every index  $\ell \in [1,n]$ is mapped to a unique pixel  $\sigma(\ell) \in \Omega$. Let $\x$ be the vector representation of  $X$ corresponding to this indexing, so that $\x_{\ell } = X_{\sigma(\ell)}$ for $\ell \in [1,n]$. Now, define the  kernel matrix $\K \in \Re^{n \times n}$ to be 
\begin{equation*}
\forall  i,j \in [1,n]: \quad \K_{ij}=\phi(\Bzeta_{\sigma(i)},\Bzeta_{\sigma(j)}),
\end{equation*}
and the diagonal matrix $\D \in \Re^{n \times n}$ to be $\D_{ii}=\sum_{j=1}^n \K_{ij}$ for $i \in [1,n]$. We can then express \eqref{neighbour} as the linear transform $\x \to \W\x$, where
\begin{equation}
\label{maineqnkernel}
\W:=\D^{-1} \K.
\end{equation}
By construction, both $\K$ and $\D$ are symmetric; $\K$ is positive semidefinite and $\D$ is positive definite \cite{milanfar2013tour}. We note that apart from image denoising \cite{buades2005non,takeda2007kernel,heide2014flexisp,nair2019hyperspectral,unni2020_pnp_registration}, kernel filters have also been used for graph signal processing \cite{yazaki2019interpolation,nagahama2021graph,gavaskar2022regularization}.

\subsection{Kernel denoiser as a proximal operator}
\label{PnPkd}

Kernel filters, and in particular NLM-type denoisers, have been shown to be very effective for  PnP regularization   \cite{sreehari2016plug,heide2014flexisp,nair2019hyperspectral,unni2020_pnp_registration}. In this section, we discuss  existing results on the optimality of the PnP reconstruction obtained using a kernel denoiser. 

Recall that in PnP, we replace the proximal operator $\mathrm{prox}_{g}$ in \eqref{ista} by a denoiser $\mathfrak{D}$ in \eqref{pnpista}. The question is whether we can express a given denoiser as a  proximal operator of some convex regularizer---convergence of the PnP iterates and optimality of the final reconstruction would immediately be settled if we can answer this in the affirmative. This question remains open for complex denoisers such as BM3D and DnCNN. However, it has been shown in \cite{nair2021fixed,gavaskar2021plug} that for the kernel filter given by \eqref{maineqnkernel}, one can indeed construct a convex regularizer $g$ such that $\W = \mathrm{prox}_{g}$.
We need some notations and definitions to explain this technical result.
 
Since $\D$ is positive definite, we can define the following inner product and norm on $\Re^n$:
\begin{equation}
\label{deninnerprod}
{\inner{\x}{\y}}_\D = \x^\top\! \D \y \quad  \text{and } \quad \norm{\x}_\D= \inner{\x}{\x}^{1/2}_\D.
\end{equation}
The latter can be viewed as a weighted $\ell_2$ norm. We associate a proximal operator with this norm.

\begin{definition} Let $g: \Re^n \to \Re \cup \{\infty\}$ be closed, proper and convex. Define $\prox_{g,\D}: \Re^n \to \Re ^n$ to be
\begin{equation}
\label{prox}
\prox_{g, \D}(\x) = \argmin_{\z \in \Re^n}\  \left\{ \frac{1}{2} \norm{\z - \x}_{\D}^2 + g(\z)\right\},
\end{equation}
We call $\prox_{g,\D}$ the scaled proximal operator of $g$. 
\end{definition}

The term ``scaled'' is used to emphasize that the weighted $\ell_2$ norm is used in \eqref{prox}. It is well-known that the objective in \eqref{prox} has a unique minimizer for the standard $\ell_2$ norm and that $\prox_{g}$ is well-defined \cite{bauschke2017convex}; it is not difficult to verify that this is also true for the weighted norm. 

The authors in \cite{nair2021fixed,gavaskar2021plug} showed that a kernel denoiser is a scaled proximal operator of a convex regularizer. 
This observation  was used to establish objective  and iterate convergence. We recall the exact formula for this regularizer which will play a central role in the rest of the discussion. 

It can be shown that the kernel filter is diagonalizable. In particular, let  $\W=\V\BLambda\V^{-1}$ be the eigendecomposition of \eqref{maineqnkernel}, where the columns of $\V$ are the eigenvectors of $\W$ and $\BLambda$ is a diagonal matrix whose first $r$ diagonal entries are in $(0, 1]$ and the remaining entries are $0$. We collect the $r$ positive eigenvalues in a diagonal matrix $\BLambda_r$ and  the  corresponding eigenvectors in $\U \in \Re^{n \times r}$. Then we can write $\W=\U\BLambda_r\U^\dagger$, where $\U^\dagger \in \Re^{r \times n}$ consists of the first $r$ rows of $\V^{-1}$. The authors in \cite{gavaskar2021plug} showed that $\W$ is a scaled proximal operator of the extended-valued convex function 
\begin{equation}
\label{regularizerold}
\displaystyle \Phi_{\W}(\x) =  \begin{cases}
     \dfrac{1}{2} \x^\top {\U^\dagger}^\top (\BLambda_r^{-1} - \I)\U^{\dagger} \x, & \mathrm{for } \ \x \in \R(\W).\\[5pt]
        \infty, & \mathrm{otherwise},          
\end{cases}
\end{equation}
Note that $\Phi_{\W}$ is completely determined by the kernel denoiser $\W$. Its domain is restricted to $ \R(\W)$ and this in effect forces the reconstruction to be in $ \R(\W)$. 

\section{Kernel Regularization}
\label{kr}
\subsection{PnP regularization using kernel denoiser}
\label{qr}

We first propose an alternate formulation for the regularizer  in \eqref{regularizerold} and prove that the scaled proximal operator of \eqref{regularizerold} is  indeed $\W$. Apart from keeping the account self-contained, our formulation makes the proofs simpler and helps in analyzing the optimization problem in \eqref{opt} with $g=\Phi_{\W}$.

Note that $\W$ is not symmetric, but it can be written as 
\begin{equation*}
\W=\D^{-1/2} \S \D^{1/2}, \quad \S:=\D^{-1/2} \K \D^{-1/2}.
\end{equation*}
Thus, $\W$ is similar to the symmetric matrix  $\S$ and is hence diagonalizable with real eigenvalues. More specifically, if $\S=\Q \BLambda \Q^\top\!$ is an eigendecomposition of $\S$, then we can write
\begin{equation*}
\W = \M \BLambda \M^{-1}, \quad   \M:=\D^{-1/2}\Q.
\end{equation*}
Define $\W^\dagger=\M \BLambda^\dagger \M^{-1}$, where $\BLambda^{\dagger}$ is diagonal and $\BLambda^{\dagger}_{ii} =  \BLambda_{ii}^{-1}$ if $\BLambda_{ii} > 0$, and $\BLambda^{\dagger}_{ii} =0$ otherwise. We remark that  $\W^\dagger$ is the generalized inverse of $\W$ and that $\W = \W \W^{\dagger} \W$ \cite{moore1920reciprocal}.

\begin{definition} Let $\W$ be a kernel denoiser. Define the  function $\Phi_{\W}: \Re^n \to \Re \cup \{\infty\}$ to be
\begin{equation}
\label{regularizer}
\displaystyle \Phi_{\W}(\x) =  \begin{cases}
     \dfrac{1}{2} \x^\top \D (\I - \W)\W^{\dagger} \x, & \mathrm{for } \ \x \in \R(\W),\\[5pt]
        \infty, & \mathrm{otherwise}.            
\end{cases}
\end{equation}
We call $\Phi_{\W}$ the kernel regularizer associated with $\W$.
\end{definition}

The extended-valued function $\Phi_{\W}$ is quadratic on its domain. Moreover, we have the following properties.
\begin{proposition}
\label{prop1}
$\Phi_{\W}$  is nonnegative, closed (epigraph is closed in $\Re^{n+1}$), proper (domain is nonempty), and convex.
\end{proposition}

In particular, it follows from Proposition \ref{prop1} that $\prox_{\Phi_{\W}, \D}$ is well-defined. 
This brings us to the key result---the characterization of a kernel denoiser as a proximal operator.
\begin{theorem}
\label{proximalmap}
$\W$ is a scaled proximal operator of $\Phi_{\W}$, i.e., 
\begin{equation*}
\forall \ \x \in \Re^n: \quad \W\x = \prox_{\Phi_{\W}, \D}(\x) .
\end{equation*}
\end{theorem}

We now relate the above results to the PnP updates in \eqref{pnpista} and \eqref{pnpadmm} where the denoiser is $\mathfrak{D}=\W$. In this case, it follows from Theorem \ref{proximalmap} that we can write \eqref{pnpista} as
\begin{equation}
\label{pnpista-prox}
\x_{k+1} = \prox_{\Phi_{\W},\D} \big(\x_k- \rho^{-1} \nabla \!f(\x_k)\big).
\end{equation}
This looks like standard ISTA \cite{beck2009fast,bauschke2017convex}, except that we have the scaled proximal operator instead of the standard proximal operator. In particular, if we replace $\nabla \!f$ (gradient of $f$ w.r.t. $\ell_2$ norm) in \eqref{pnpista-prox} with $\D^{-1}\nabla \!f$ (gradient of $f$ w.r.t. weighted-$\ell_2$ norm), then PnP updates amount to minimizing $f+\rho \Phi_{\W}$ \cite{nair2021fixed,gavaskar2021plug}. Similarly, we can write the PnP-ADMM updates as 
\begin{align}
\label{pnpadmm-prox}
& \x_{k+1} = \underset{\x}{\mathrm{argmin}} \ f(\x) + \frac{1}{2 \rho}\| \x - (\v_k - \z_k) \|_\D^2, \\
& \v_{k+1} = \prox_{\Phi_{\W},\D} (\x_{k+1} + \z_k), \nonumber \\ 
& \z_{k+1} = \z_k + (\x_{k+1} - \v_{k+1}), \nonumber
\end{align}
Again, this looks like standard ADMM except that the weighted $\ell_2$ norm is used instead of the standard $\ell_2$ norm for the $\x$ update.  In summary, if we use a kernel denoiser for PnP regularization and we make some minor algorithmic adjustments to the PnP updates, the resulting iterates are guaranteed to converge to the minimum of $f+ \rho \Phi_{\W}$. Thus, the seemingly ad-hoc idea of ``plugging'' a denoiser into the reconstruction process does amount to solving a regularization problem if we use a kernel denoiser. 

Having been able to associate an optimization problem with the PnP iterations, a natural question is whether we can solve the optimization problem more efficiently. This would provide an alternate means for performing PnP regularization. We next show that this is indeed the case for linear inverse problems.

\subsection{Reduction to a linear system}
\label{reduction}

Recall from the previous discussion that PnP regularization using a kernel denoiser amounts to solving \eqref{opt}, where $f$ and $g$ are given by \eqref{loss}. In other words, the optimization problem associated with PnP-ISTA and PnP-ADMM is 
\begin{equation}
\label{pnpopt}
\underset{\x \in \R(\W)}{\mathrm{min}} \ \left\{  \frac{1}{2} \| \y - \F\x \|^2_2 + \frac{\rho}{2} \x^\top \D (\I - \W) \W^{\dagger} \x \right\}. 
\end{equation}

It is not obvious if \eqref{pnpopt} has a minimizer since $\F$ has a non-trivial null space \cite{boyd2004convex}. 
In \cite{nair2021fixed,gavaskar2021plug}, the authors have assumed the existence of a minimizer without proof. We now establish this  rigorously. 
\begin{proposition}
\label{minimizers}
The optimization problem \eqref{pnpopt} is solvable. In particular, consider the linear system  $\A\z=\b$, where 
\begin{equation}
\label{linearsystem}
\A=\W^\top \F^\top \F \W + \rho \W^\top \D (\I - \W), \ \  \b= \W^\top\F^\top \y.
\end{equation}
Then $\A\z=\b$ is solvable. Moreover, $\x^\ast$ is a minimizer of  \eqref{pnpopt} if and only if $\x^\ast=\W \z$ where $\z$ is a solution of $\A\z=\b$.

\end{proposition}

We remark that although $\W^{\dagger}$ appears in \eqref{pnpopt}, it does not come up in \eqref{linearsystem}. This is important since computing $\W^{\dagger}$ from $\W$ via its eigendecomposition would be impractical given the black-box nature of $\W$ and its size. On the other hand, note that $\W^\top = \D\W\D^{-1}$. Hence, computing $\W^\top\x$ has the same cost as that of computing $\W\x$  since $\D$ is a diagonal matrix. 

Though the linear system $\A\z=\b$ is solvable, $\A$ can nevertheless be singular. However, most iterative linear solvers in their original form (e.g. Krylov methods) require $\A$ to be nonsingular for theoretical convergence  \cite{simoncini2005occurrence}. The following proposition is useful in this regard.
\begin{proposition}
\label{uniqueness}
Let $\e \in \Re^n$ be the all-ones vector. If $\W$ is nonsingular and $\F\e \neq \mathbf{0}$, then $\A$ is nonsingular.  
\end{proposition}

For deblurring, superresolution and inpainting, $\F$ trivially satisfies the hypothesis in Proposition \ref{uniqueness} since $\e$ does not belong to the null space of $\F$. On the other hand, for bilateral and NLM filters, $\W$ is nonsingular if the spatial component of the kernel function $\phi$ is a hat function \cite[Theorem $3.13$]{nair2021fixed}. 

Note that if $\W$ is nonsingular, then $\A\z=\b$ can equivalent be expressed as $\C\z=\d$, where
\begin{equation}
\label{linearsystem2}
\C= \F^\top \F \W + \rho \D (\I - \W) \quad \text{and} \quad \d= \F^\top \y.\end{equation}
The advantage with \eqref{linearsystem2} compared to \eqref{linearsystem} is that $\W^\top$ does not appear in $\C$. However, unlike $\A$, $\C$ is no longer guaranteed to be symmetric.  

\begin{figure*}
\centering
\subfloat{\includegraphics[width=0.24\linewidth]{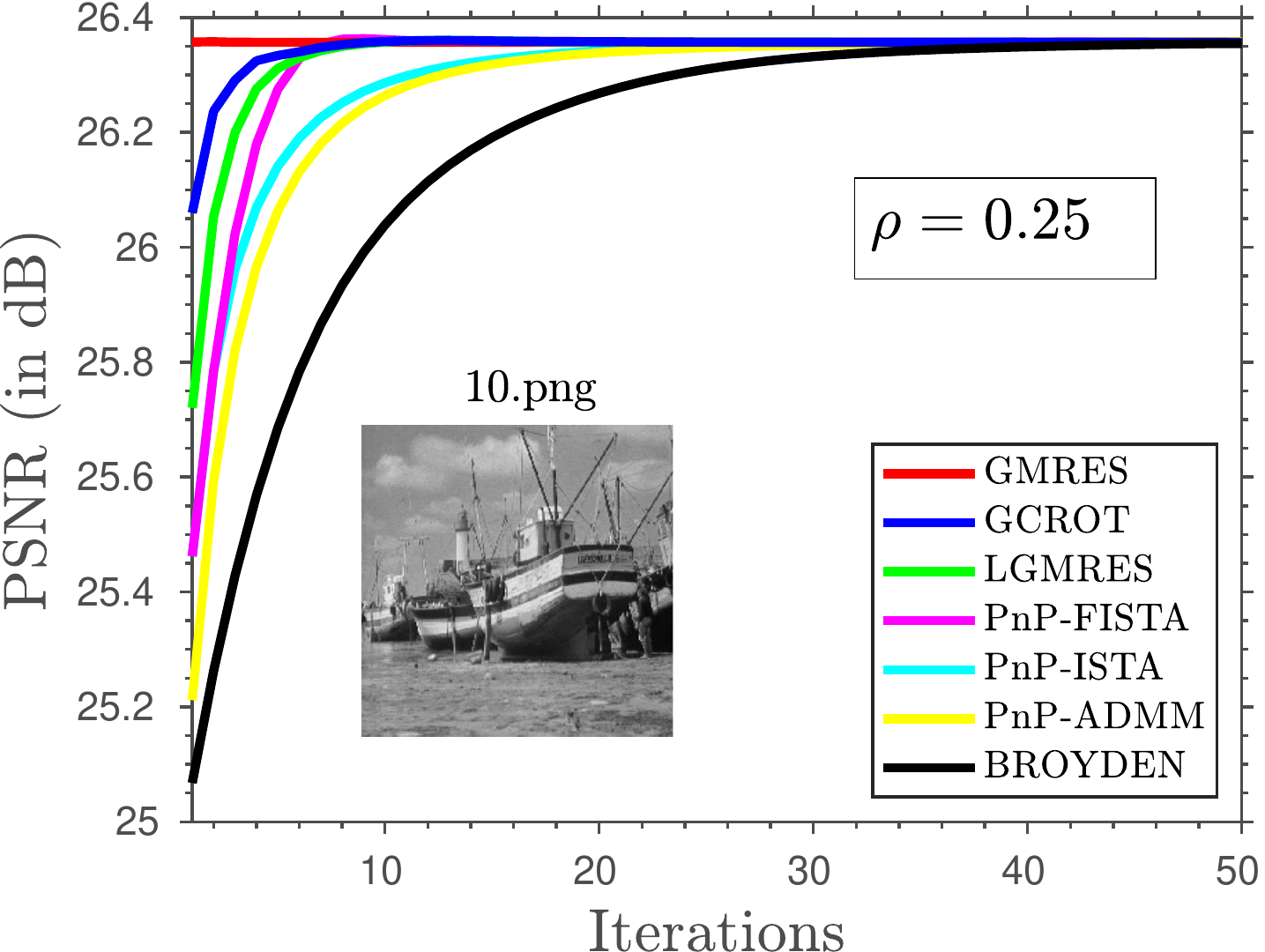}}
\hspace{0.1mm}
\subfloat{\includegraphics[width=0.24\linewidth,height=0.181\linewidth]{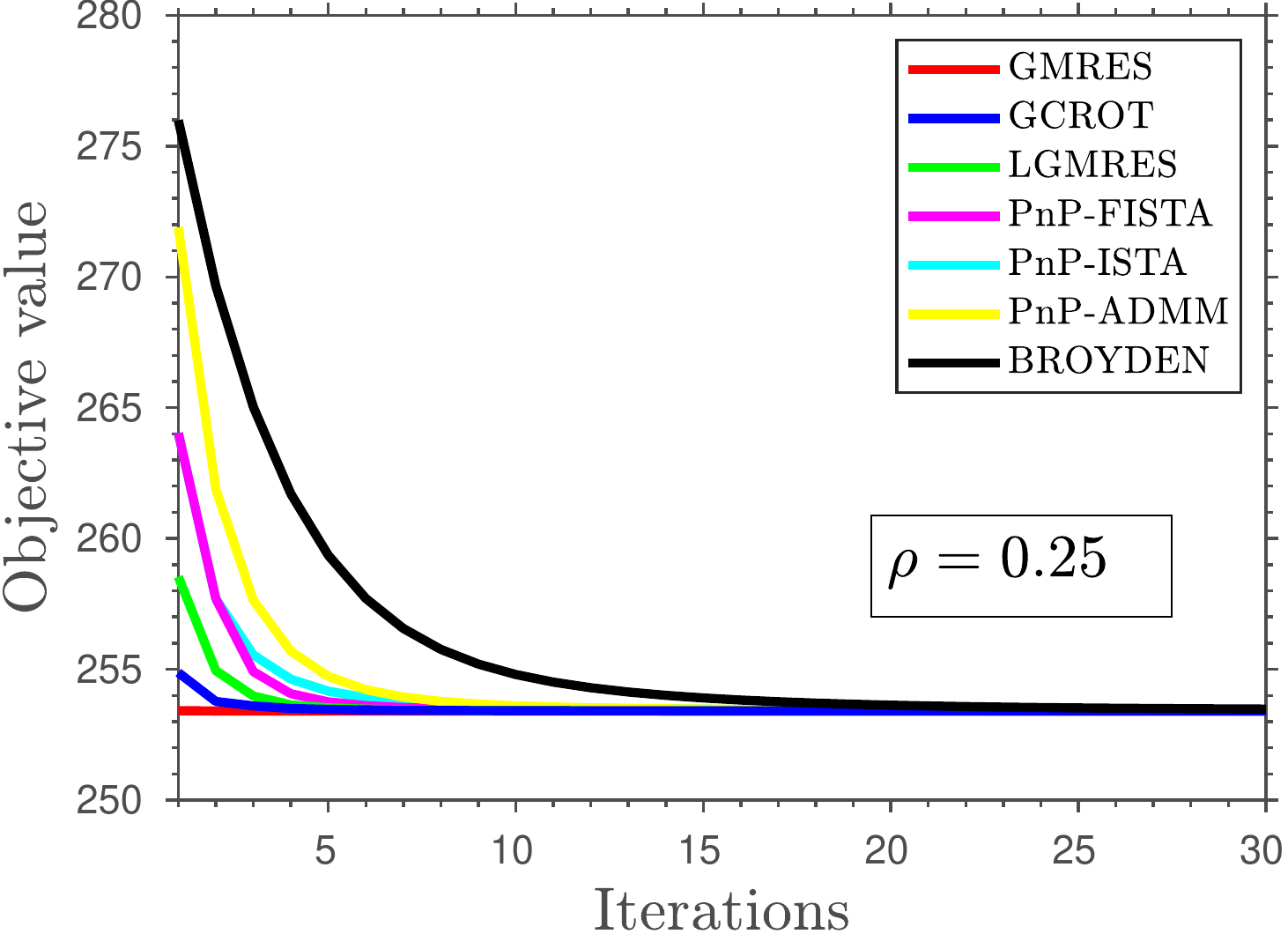}} 
\subfloat{\includegraphics[width=0.24\linewidth,height=0.181\linewidth]{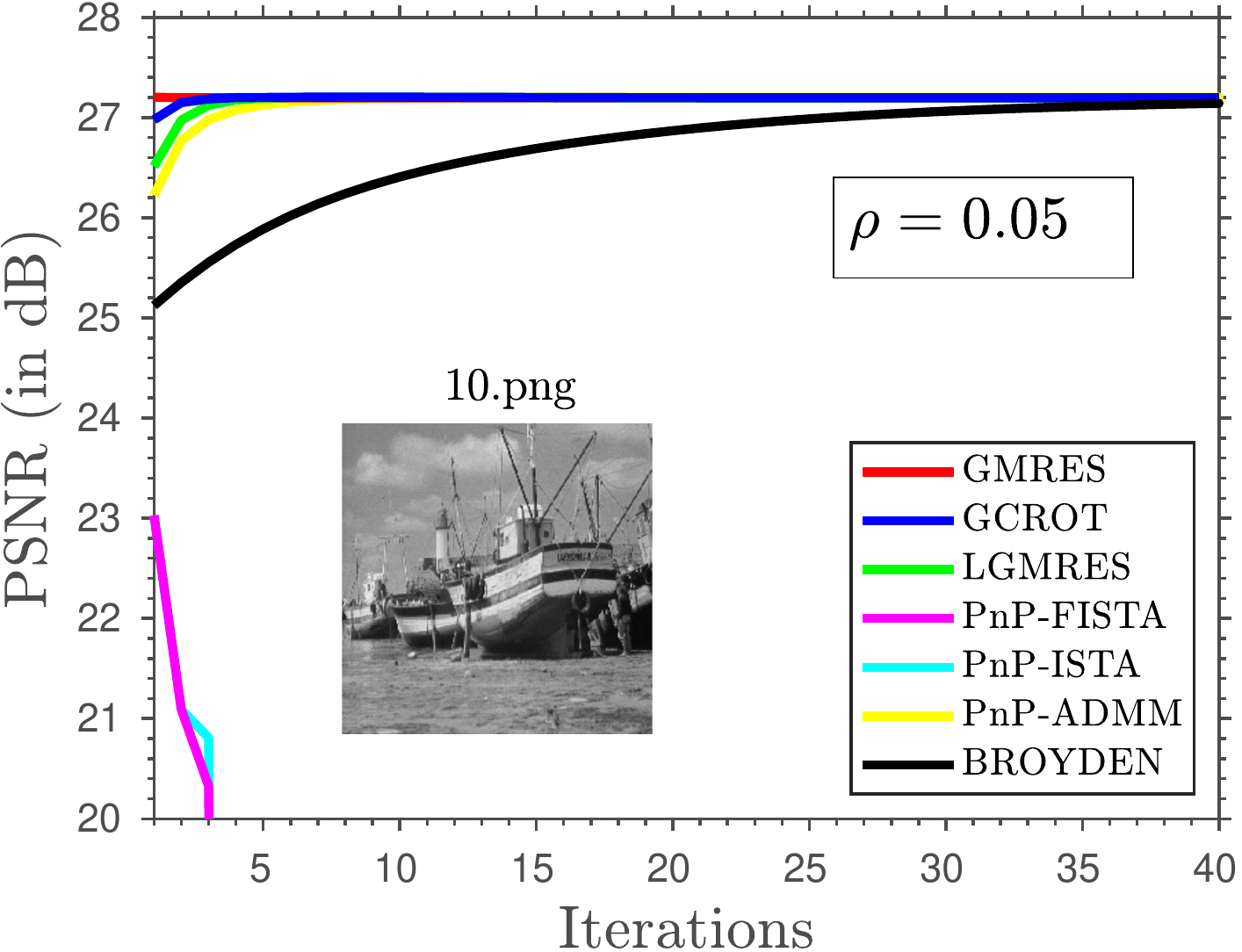}} 
\hspace{0.1mm}\subfloat{\includegraphics[width=0.24\linewidth]{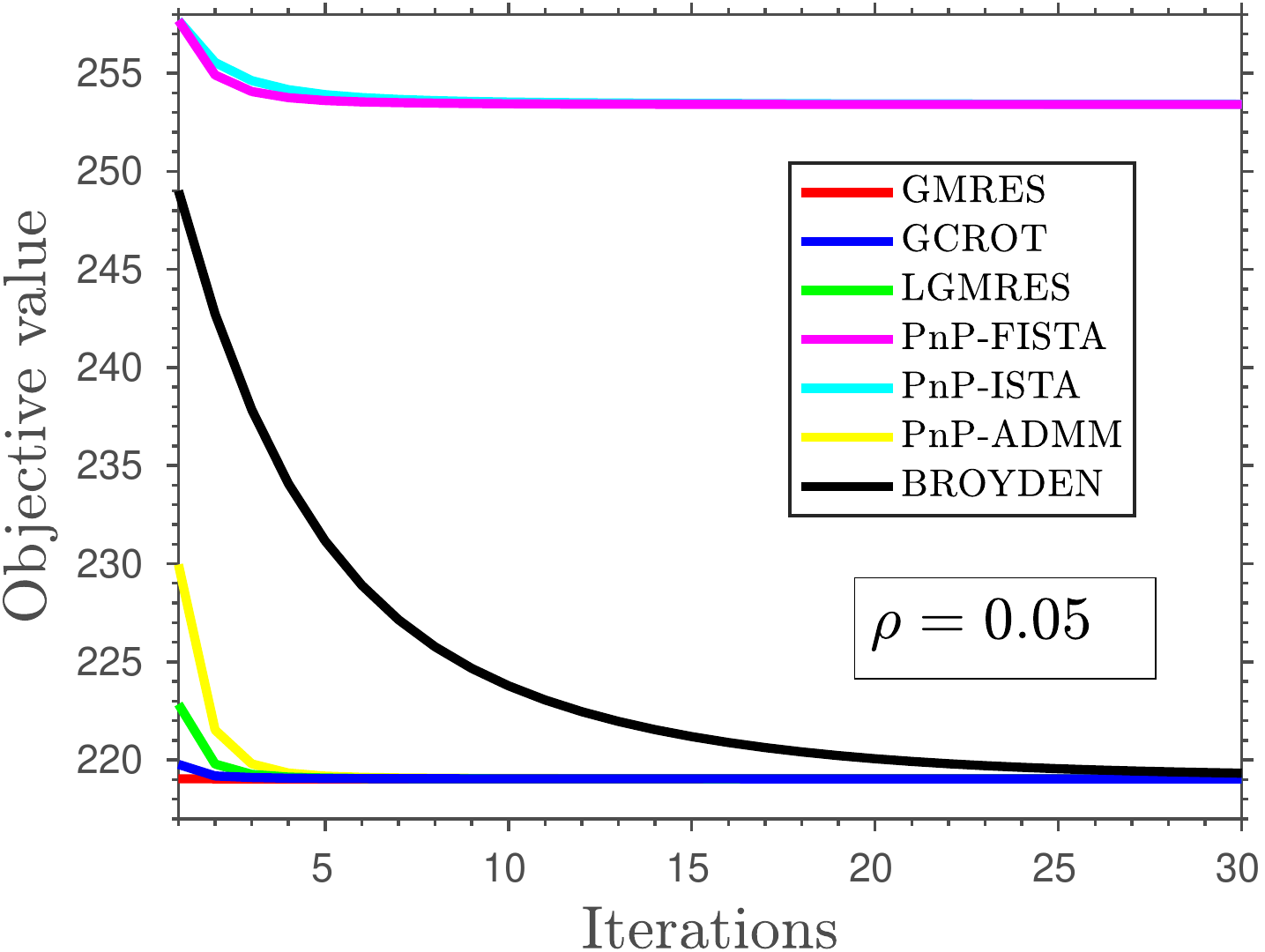}} 
\caption{Comparison of the efficiency of iterative solvers for image deblurring. A Gaussian  blur of $25 \times 25$ and standard deviation $1.6$ is used and the noise level is $\sigma=10/255$. The plots show the evolution of  PSNR and objective function  \eqref{pnpopt} for $\rho = 0.25$ (left column) and  $\rho = 0.05$ (right column). We have used NLM as the denoiser. The comparison is between PnP-ISTA, PnP-FISTA and PnP-ADMM \cite{nair2021fixed,gavaskar2021plug} and  the proposed linear solvers  GMRES, GCROT and LGMRES. As expected, our  solvers exhibit faster convergence than the iterative minimization of \eqref{pnpopt} using PnP-ISTA, PnP-FISTA and PnP-ADMM. For example, GCROT requires just $40\%$ percent of the iterations taken by PnP-ADMM to stabilize. Note that PnP-ISTA and PnP-FISTA diverges when $\rho =0.05$.}
\label{Empconv}
\end{figure*}

\begin{figure}
\centering
\subfloat{\includegraphics[width=0.48\linewidth]{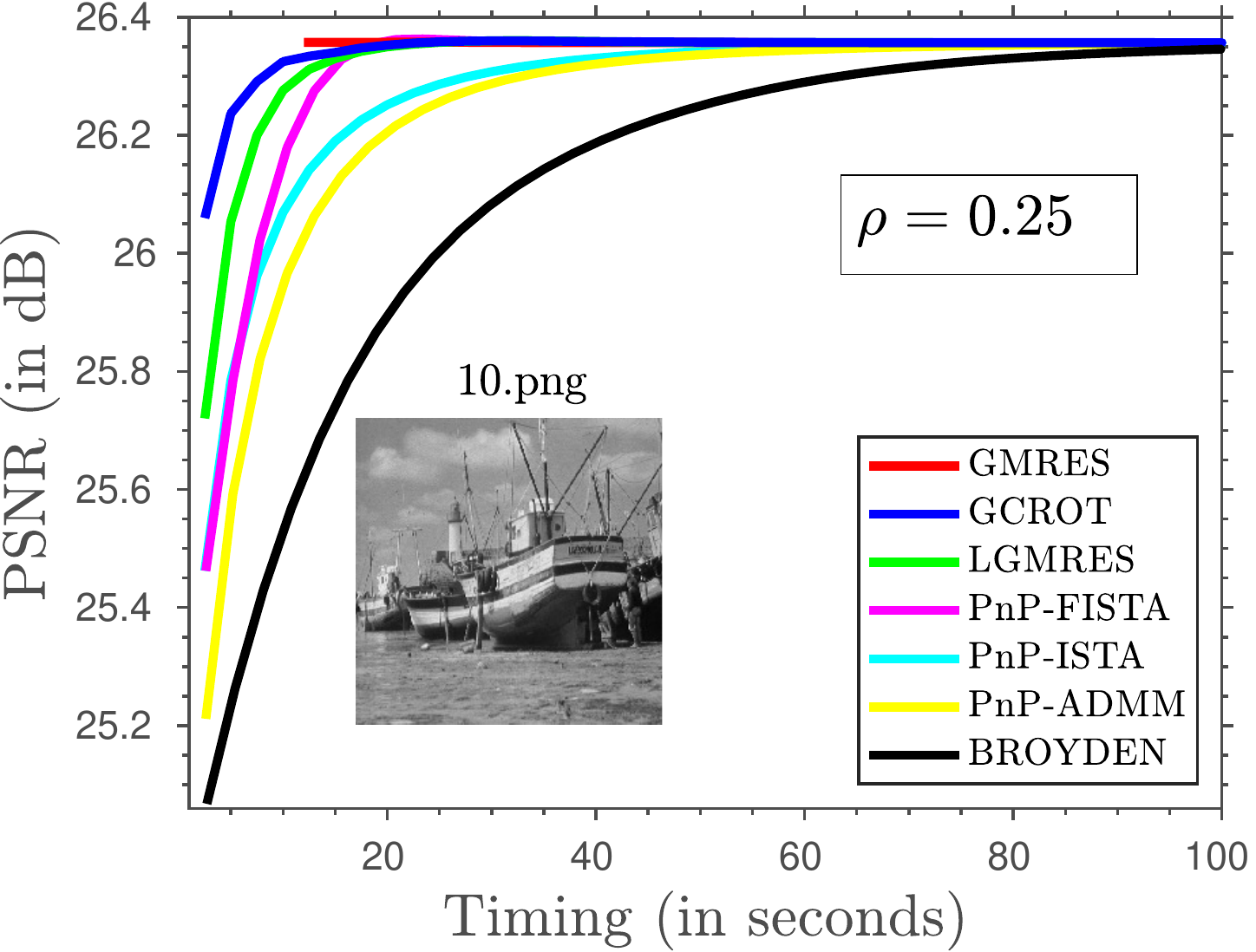}}
\hspace{0.1mm}
\subfloat{\includegraphics[width=0.48\linewidth]{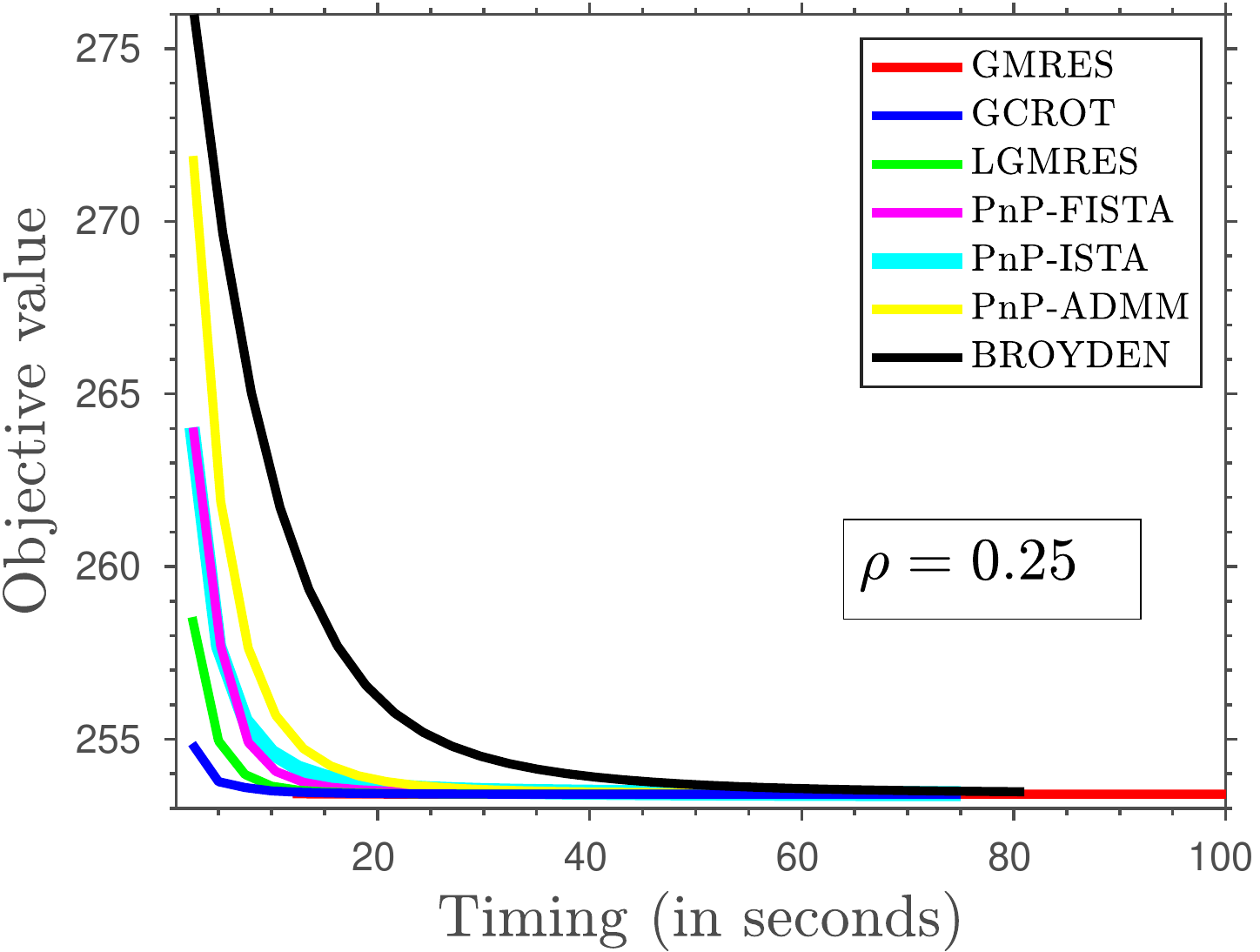}} 
\caption{Efficiency comparison of the algorithms for the experiment in Fig.~\ref{Empconv}. The linear solvers are faster to stabilize than the iterative minimization of \eqref{pnpopt} using PnP algorithms.}
\label{Empconvtime}
\end{figure}

\subsection{Proposed algorithm}
\label{propsolver}

\begin{algorithm}
\KwIn{Measurements $\y$ and forward operator $\F$.} 
\KwOut{Solution of \eqref{pnpopt}.}
Form the guide image $\u$ from measurements $\y$\; \label{surrogate}
Construct the denoising operator  $\W$ from $\u$\; 
Use $\W$ and $\F$ to construct operator $\C$\;
Solve $\C\z =\F^\top\y$ \label{oursolver}\;
Return $\W\z$.
\caption{Regularization using a linear solver}
\label{propalgo}  
\end{algorithm}

The overall algorithm stemming from the previous discussion is summarized in Algorithm \ref{propalgo}. In the first step, the guide image $\u$ is derived from the measurements $\y$ using about $5\mbox{-}10$ PnP iterations, as done in \cite{sreehari2016plug,Teodoro2019PnPfusion,nair2021fixed,gavaskar2021plug}. This is used to construct the kernel denoiser $\W$ and in turn the operator $\C$. We choose NLM as the  preferred kernel denoiser because of its superior regularization capabilities, where the features are derived from the guide image $\u$ as explained in Section \ref{kd}. The final and core step is the solution of $\C\z =\F^\top\y$. We have used  iterative solvers such as LGMRES and GCROT  (using \textit{linalg} module in the \textit{scipy} package). 
which seem to provide the best tradeoff between iteration complexity and the number of iterations required to converge. This is  highlighted in Fig.~\ref{Empconv}. 

At this point, we wish to clarify that $\W$, $\F$ and $\F^\top$ in \eqref{linearsystem2} are not stored as matrices. Rather, they are implemented efficiently as input-output black boxes. Indeed, $\W\x$ can be computed directly from \eqref{neighbour}. For deblurring, $\F$ and $\F^\top$ are implemented using spatial convolution (narrow blur) or FFT (wide blur). For superresolution, $\F$ is implemented using spatial convolution followed by decimation and $\F^\top$ is implemented using upsampling followed by convolution. For inpainting, computing $\F$ (and $\F^\top=\F$) amounts to sampling the observed pixels.

\begin{figure*}
\centering
  \subfloat{\includegraphics[width=0.105\linewidth]{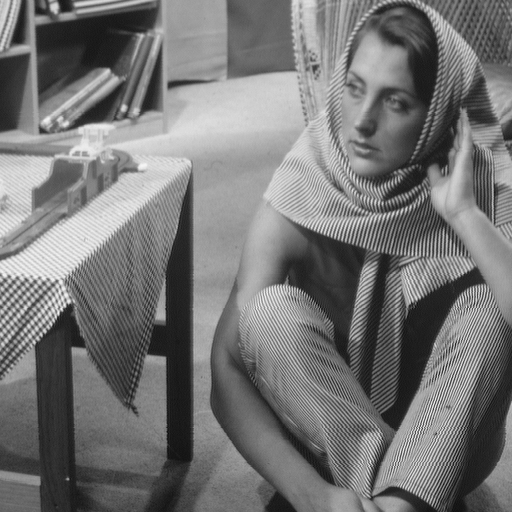}} \hspace{0.01mm}
  \subfloat{\includegraphics[width=0.105\linewidth]{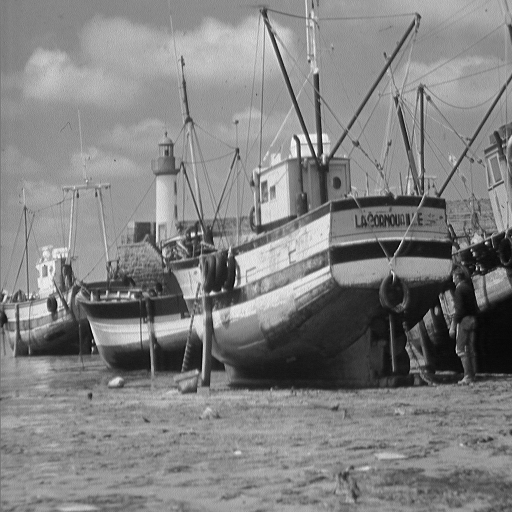}} \hspace{0.01mm}
  \subfloat{\includegraphics[width=0.105\linewidth]{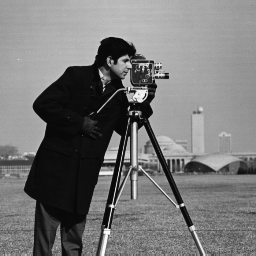}} \hspace{0.01mm}
  \subfloat{\includegraphics[width=0.105\linewidth]{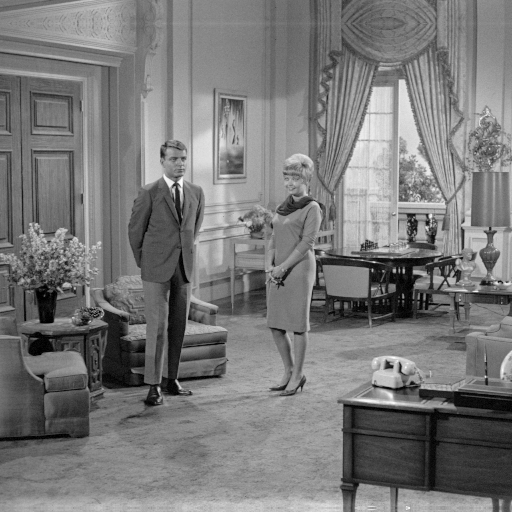}} \hspace{0.01mm}
  \subfloat{\includegraphics[width=0.105\linewidth]{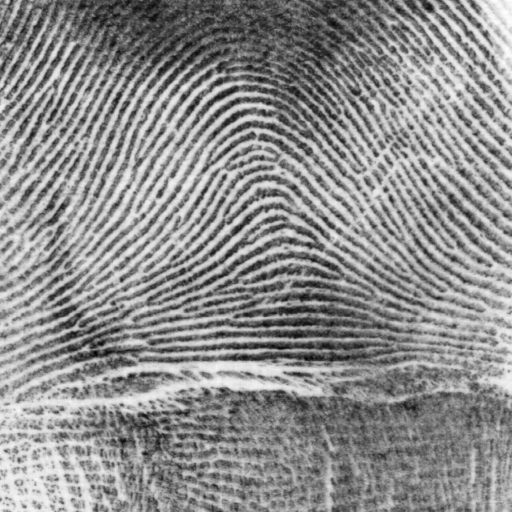}} \hspace{0.01mm}
  \subfloat{\includegraphics[width=0.105\linewidth]{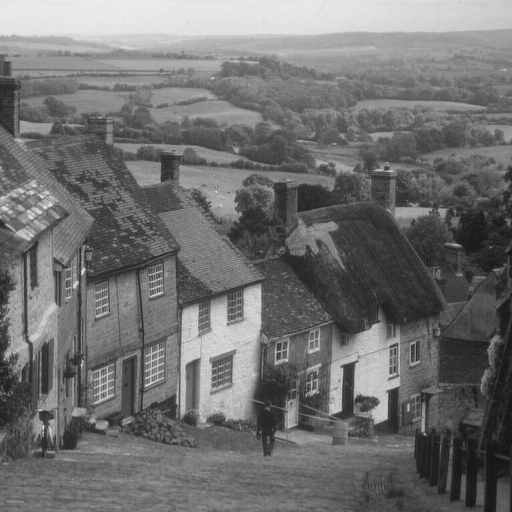}} \hspace{0.01mm}
  \subfloat{\includegraphics[width=0.105\linewidth]{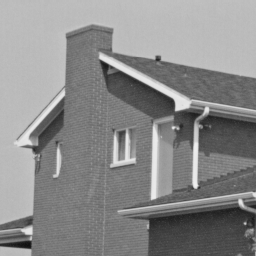}} \hspace{0.01mm}
  \subfloat{\includegraphics[width=0.105\linewidth]{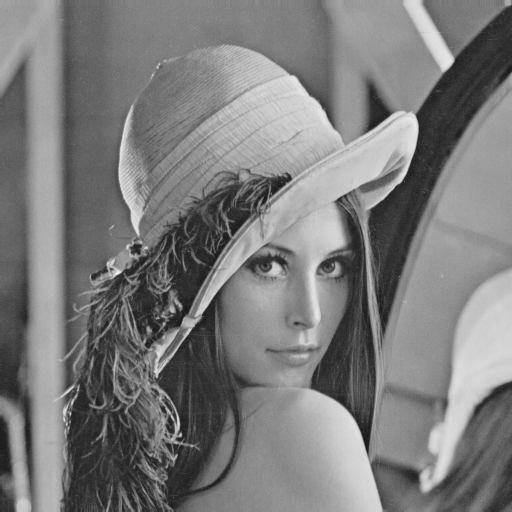}} \hspace{0.01mm}
  \subfloat{\includegraphics[width=0.105\linewidth]{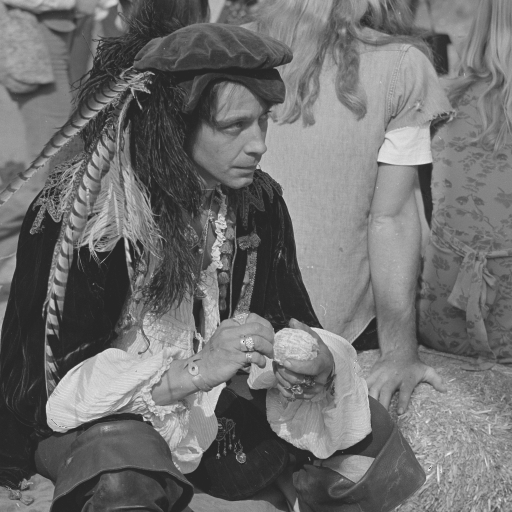}} 
   \caption{Grayscale input images used in the experiments in this paper.}
\label{Inputimages}
\end{figure*}

\begin{table*}[!htp]
\caption{Comparison of the PnP algorithms in \cite{nair2021fixed,gavaskar2021plug} and the proposed solver.}
\centering
\begin{tabular}{ | p{3.5cm}  | p{4.5cm} | p{4cm} | p{4.5cm} |}
\hline
\textbf{Property} & \textbf{PnP algorithms \cite{nair2021fixed,gavaskar2021plug}} & \textbf{Proposed} & \textbf{Comments} \\
\hline
convergence rate & sublinear convergence & superlinear convergence 
& superior convergence rate of our method  is empirically seen in Fig.~\ref{Empconv}. \\
\hline
cost per iteration & more for PnP-ADMM in most applications since we need to solve a linear system  & same as PnP-ISTA & total time taken to stabilize is less for our method as shown in Fig.~\ref{Empconvtime}.\\
\hline
flexibility of step size $\rho$ & $ \rho \geqslant 0.5\sigma_{\mathrm{max}}(\D^{-1/2}\F^\top\F\D^{-1/2})$  to guarantee convergence for PnP-ISTA & convergent for any $\rho>0$& as seen in Fig.~\ref{Empconv}, $\rho$ needs to be small to obtain state-of-the-art results. \\
\hline
existence of minimizer  of \eqref{pnpopt} & assumed & proved in Proposition \ref{minimizers} &  we can certify this for general linear inverse problems.\\
\hline
unique minimizer of \eqref{pnpopt} & not discussed & proved in Proposition \ref{uniqueness} &  uniqueness can be certified for deblurring, inpainting and superresolution.  \\
\hline
\end{tabular}
\label{compTIP21andourpaper}
\end{table*}

\begin{table*}[!htp]
\caption{Comparison of PSNR using different number of PnP iterations to obtain the guide image from the measurements.}
\centering
\begin{tabularx}{\textwidth}{ |c|*{6}{Y} || *{6}{Y|} }
\hline
&  \multicolumn{6}{c||}{Deblurring} & \multicolumn{6}{c|}{Inpainting} \\
\hline
\# PnP iterations & $0$ & $1$ & $2$ & $3$ & $4$ & $5$ & $0$ & $1$ & $2$ & $3$ & $4$ & $5$\\
\hline
\scriptsize \textit{barbara} & $\textbf{23.88}$ & $23.81$ & $23.88$ & $23.82$ & $23.81$ & $23.79$ & $24.16$ & $23.10$ & $23.58$ & $24.81$ & $26.11$ & $\textbf{27.19}$\\
\scriptsize \textit{boat} & $27.22$ & $27.36$ & $\textbf{27.46}$ & $27.30$ & $27.22$ & $27.11$ & $26.93$ & $27.47$ & $27.57$ & $27.70$ & $27.87$ & $\textbf{27.99}$ \\
\scriptsize \textit{cameraman} & $24.59$ & $24.67$ & $\textbf{24.78}$ & $24.68$ & $24.62$ & $24.54$ & $23.55$ & $24.29$ & $24.48$ & $24.64$ & $24.80$ & $\textbf{24.93}$ \\
\scriptsize \textit{couple} & $26.74$ & $26.92$ & $\textbf{27.00}$ & $26.95$ & $26.89$ & $26.82$ & $26.81$ & $27.09$ & $27.21$ & $27.45$ & $27.72$ & $\textbf{27.90}$ \\
\scriptsize \textit{fingerprint} & $\textbf{25.88}$  & $25.78$ & $25.27$ & $24.73$ & $24.36$ & $24.00$ & $23.52$ & $25.14$ & $25.36$ & $25.63$ & $25.94$ & $\textbf{26.20}$\\
\scriptsize \textit{hill} & $28.16$ & $28.15$ & $\textbf{28.27}$ & $28.15$ & $28.15$ & $28.13$ & $28.86$ & $28.90$ & $28.96$ & $29.06$ & $29.18$ & $\textbf{29.27}$\\
\scriptsize \textit{house} & $\textbf{29.96}$ & $29.48$ & $29.62$ & $29.37$ & $29.92$ & $29.14$ & $29.34$ & $30.48$ & $30.86$ & $31.32$ & $31.82$ & $\textbf{32.17}$\\
\scriptsize \textit{lena} & $30.32$ & $30.24$ & $\textbf{30.43}$ & $30.25$ & $30.19$ & $30.08$ & $30.35$ & $31.05$ & $31.23$ & $31.46$ & $31.67$ & $\textbf{31.84}$ \\
\scriptsize \textit{man} & $28.06$ & $28.06$ & $\textbf{28.17}$ & $28.05$ & $28.00$ & $27.91$ & $28.16$ & $28.39$ & $28.48$ & $28.62$ & $28.77$ & $\textbf{28.87}$\\
\hline
\end{tabularx}
\label{PnPiter1}
\end{table*}

\subsection{Comparison with existing PnP algorithms} 
\label{advoverpnp}

We remark that the optimization in \eqref{pnpopt} can be solved using either existing PnP algorithms \cite{nair2021fixed,gavaskar2021plug} or the proposed  linear solver  \eqref{linearsystem2}. Moreover, both are iterative in nature. An obvious question is what advantage does the latter offer? In this regard, we wish to discuss the following points.

\begin{itemize}
\item  For solving linear systems, Krylov solvers like GMRES, GCROT and LGMRES are known to be superlinearly convergent \cite{simoncini2005occurrence}.
On the other hand, iterative PnP algorithms such as PnP-ISTA, PnP-FISTA and PnP-ADMM are only sublinearly convergent \cite{beck2009fast,he20121}. Hence, in theory, Algorithm \ref{propalgo} with Krylov solvers  should require fewer iterations as compared to the PnP algorithms in \cite{nair2021fixed,gavaskar2021plug}.  We have experimented using Krylov solvers GMRES \cite{saad1986gmres}, LGMRES \cite{baker2005technique}, GCROT \cite{hicken2010simplified} and the quasi-Newton Broyden solver  \cite{broyden1965class,more1976global}. The evolution of  the objective function and PSNR for a deblurring experiment is shown in Fig.~\ref{Empconv} for  $\rho \in \{0.05,0.25\}$. Compared to PnP-ISTA, PnP-FISTA and PnP-ADMM algorithms  \cite{nair2021fixed,gavaskar2021plug}, Algorithm \eqref{propalgo} indeed converges much faster to a minimizer of $f+\rho \Phi_{\W}$, regardless of the linear solver used. Based on the convergence rate, we can order them as follows: GMRES ($12$s) $>$ GCROT ($2.5$s) $>$ LGMRES ($2.5$s) $>$ PnP-ISTA, PnP-FISTA, PnP-ADMM ($2.5$s) $>$ Broyden ($2.7$s), where the per-iteration cost is mentioned within brackets. We note that the per-iteration cost of GCROT and LGMRES is comparable to that of the PnP algorithms in \cite{nair2021fixed,gavaskar2021plug} but they converge much faster.  

\item  For our proposal, we are required to solve just one linear system. On the contrary, in every iteration of PnP-ADMM \eqref{pnpadmm-prox}, the following linear system needs to be solved as part of the $\x$ update (outer iterations):
\begin{equation*}
(\F^\top\F + \D)\x_{k+1} = \F^\top \y + \D(\v_k - \z_k).
\end{equation*}
This in turn requires an iterative solver  (inner iterations) for deblurring and superresolution. 

\item In Fig.~\ref{Empconv}, notice that PnP-ISTA and PnP-FISTA diverges when $\rho$ is $0.05$. This is possibly because this choice of $\rho$  violates the bound $\rho \geqslant 0.5 \ \sigma_{\mathrm{max}}(\D^{-1/2}\F^\top\F\D^{-1/2})$ that is used to guarantee convergence \cite{nair2021fixed,gavaskar2021plug}. On the other hand, our linear solver converges for any $\rho >0$. This is important since we often obtain better reconstructions for smaller $\rho$, e.g., PSNR of $27.2$ dB at $\rho=0.05$ compared to $26.35$ dB at $\rho=0.25$. 
\end{itemize} 

In summary, Algorithm \ref{propalgo} has several computational advantages over the PnP algorithms in \cite{nair2021fixed,gavaskar2021plug}, though they solve the same optimization problem. A detailed comparison is provided in Table \ref{compTIP21andourpaper}.

\subsection{Comparison of convergent linear solvers}
We consider four Krylov techniques for solving \eqref{linearsystem2}: GMRES, LGMRES, GCROT, and quasi-Newton Broyden. They come with convergence guarantees and can handle asymmetric systems. We have already shown in Section \ref{advoverpnp} that Algorithm \ref{propalgo} has a better convergence rate than PnP algorithms. Further, the evolution of  the objective function (and PSNR) for a deblurring experiment is shown in Figures \ref{Empconv} and \ref{Empconvtime}. Based on convergence rate (number of iterations), we can order them as follows: GMRES $>$ GCROT $>$ LGMRES $>$ Broyden, whereas based on the time taken to stabilize, the ordering is: GCROT $>$ LGMRES $>$ GMRES $>$ Broyden. We note that all four solvers solve the same optimization problem and hence stabilize to the same PSNR (objective) value. It is clear from the empirical analysis that GCROT offers an optimal tradeoff between convergence rate and cost per iteration.

\begin{figure*}[!htp]
\centering
  \subfloat[Ground truth]{\includegraphics[width=0.16\linewidth]{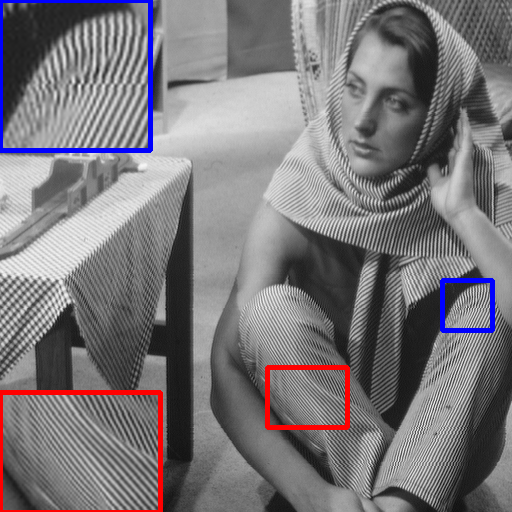}} \hspace{0.1mm}
  \subfloat[Observed]{\includegraphics[width=0.16\linewidth]{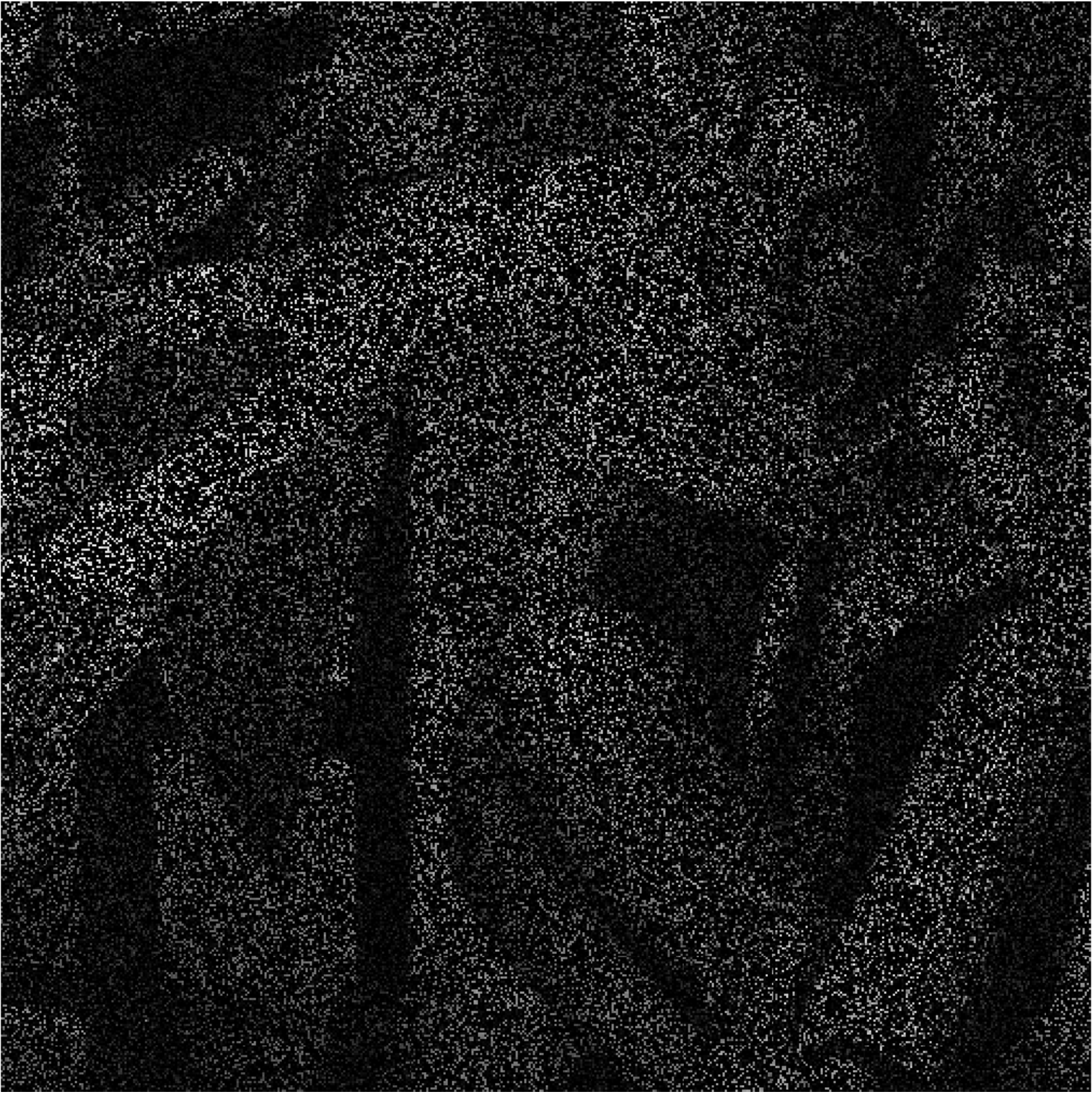}} \hspace{0.1mm}
  \subfloat[EPLL \cite{zoran2011learning}]{\includegraphics[width=0.16\linewidth]{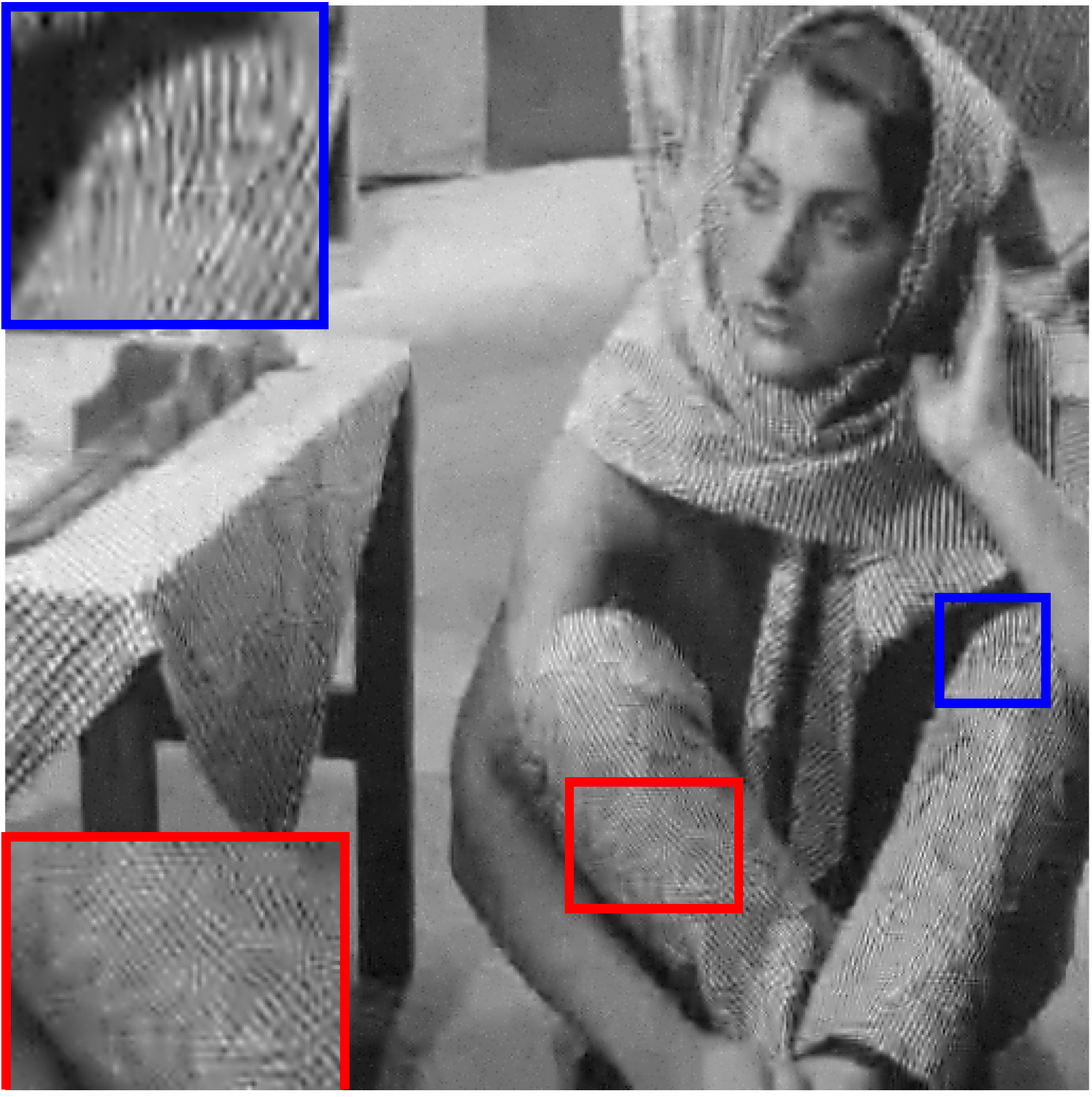}}
  \subfloat[IDBP-CNN \cite{Tirer2019_iter_denoising}]{\includegraphics[width=0.16\linewidth]{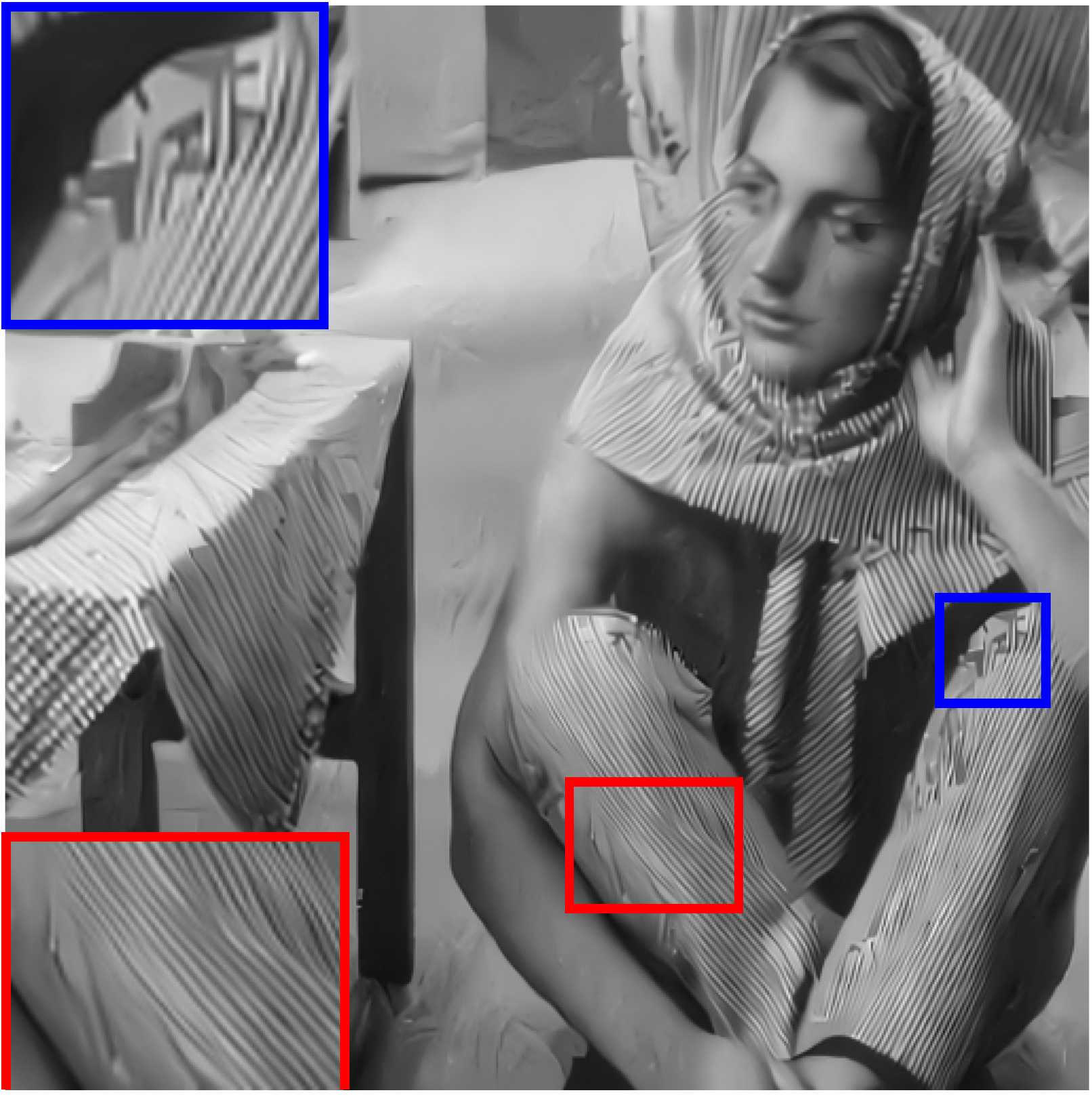}} \hspace{0.1mm} 
  \subfloat[PnP-CNN \cite{Ryu2019_PnP_trained_conv}]{\includegraphics[width=0.16\linewidth]{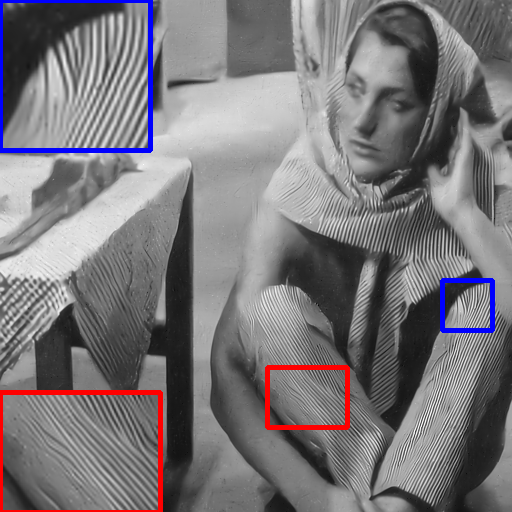}} \hspace{0.1mm}
  \subfloat[Proposed.]{\includegraphics[width=0.16\linewidth]{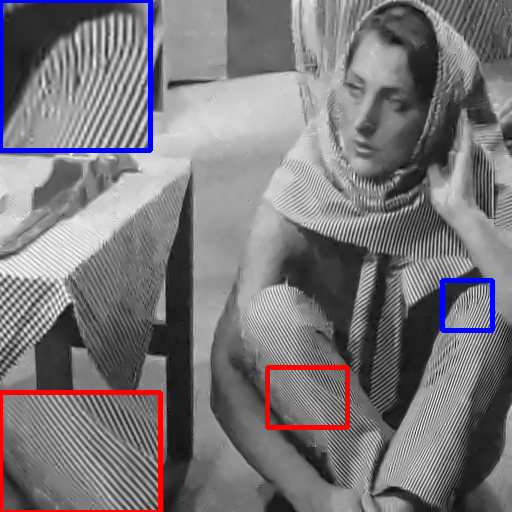}}
\caption{Image inpainting from just $20\%$ pixels at $\sigma=10/255$ noise level. The proposed method ($\rho = 0.05$) is compared with state-of-the-art methods (c)--(e). PSNR and SSIM values: (c) $24.87$ dB,  $0.729$; (d) $26.12$ dB, $0.806$; (e) $26.99$ dB, $0.812$; and (f)  $\textbf{27.67}$ dB, $\textbf{0.840}$. Comparison of the zoomed regions shows the superior reconstruction capability of our method in textured regions. }
\label{inpaintfig}
\end{figure*}

\begin{table*}[!htp]
\caption{Experiment 1: Comparison of timing, PSNR and SSIM for Inpainting with $80$ percent missing pixels (Noiseless case).}
\centering
\resizebox{\textwidth}{!}{%
\begin{tabularx}{\textwidth}{ |c|| c|| *{10}{Y|} }
\hline
Method & \scriptsize \textbf{time (s)} & \scriptsize \textit{barbara} & \scriptsize \textit{boat} & \scriptsize \textit{cameraman}  & \scriptsize \textit{couple} & \scriptsize \textit{fingerprint} &  \scriptsize \textit{hill} & \scriptsize \textit{house}  & \scriptsize \textit{lena} & \scriptsize \textit{man}\\
\hline
\multirow{ 2}{*}{EPLL \cite{zoran2011learning}} & \multirow{ 2}{*}{$258.13$} & $25.43$  $0.792$ & $27.51$  $0.788$ & $24.31$  $0.816$ & $27.69$  $0.801$ & $24.20$ $0.845$ &  $29.06$  $0.776$ & $30.71$  $0.867$  & $30,42$  $0.867$ & $27.94$ $0.795$\\
\hline
\multirow{ 2}{*}{IRCNN\cite{zhang2017learning}} & \multirow{ 2}{*}{$31.06$} & $27.34$  $0.858$ & $27.88$  $0.809$ & $25.27$  $0.838$ & $28.23$  $0.834$ & $25.97$ $0.887$ & $29.24$  $0.804$ & $ 32.21$  $0.888$ & $31.56$  $0.889$  & $28.60$ $\textbf{0.830}$ \\
\hline
\multirow{ 2}{*}{IDBP-BM3D \cite{Tirer2019_iter_denoising}} & \multirow{ 2}{*}{$45.5$} & $25.55$ $0.841$ & $28.51$  $\textbf{0.824}$ & $24.86$  $0.840$ & $\textbf{28.80}$  $0.846$ & $25.09$ $0.878$ &  $29.74$  $0.810$ & $\textbf{33.78}$  $\textbf{0.893}$ & $32.13$  $0.893$  & $28.65$ $0.824$
\\
\hline
\multirow{ 2}{*}{IDBP-CNN \cite{Tirer2019_iter_denoising}} & \multirow{ 2}{*}{$39.77$} & $24.29$  $0.796$ & $27.72$  $0.803$ & $24.24$  $0.826$ & $27.98$  $0.818$ & $25.59$ $0.865$ &  $29.01$  $0.790$ & $32.14$  $0.881$ & $31.22$  $0.880$ & $28.58$ $\textbf{0.830}$\\
\hline
\multirow{ 2}{*}{Proposed - Yaroslavsky \cite{yaroslavsky1985digital}} & \multirow{ 2}{*}{$\textbf{26.75}$} &$29.14$  $0.884$ &  $27.88$  $0.767$ & $25.10$  $0.816$ & $28.07$  $0.792$ & $26.68$ $0.913$ &  $28.96$  $0.758$ & $32.12$  $0.850$ & $31.25$  $0.853$ & $28.32$ $0.719$ \\
\hline
\multirow{ 2}{*}{Proposed - Bilateral \cite{tomasi1998bilateral}} & \multirow{ 2}{*}{$\textbf{26.75}$} &$29.37$  $0.896$ &  $28.24$  $0.817$ & $24.85$  $0.842$ & $28.56$  $0.840$ & $26.34$ $0.902$ &  $29.53$  $0.823$ & $32.83$  $0.829$ & $31.98$  $0.894$ & $28.80$ $0.819$ \\
\hline
\multirow{ 2}{*}{Proposed - NLM (Laplacian)} & \multirow{ 2}{*}{$\textbf{28.15}$} &$29.32$  $0.890$ &  $28.28$  $0.821$ & $25.07$  $0.828$ & $28.64$  $0.845$ & $26.29$ $0.900$ &  $29.75$  $0.823$ & $33.43$  $0.895$ & $32.06$  $0.894$ & $28.82$ $0.829$ \\
\hline
\multirow{ 2}{*}{Proposed - NLM (Gaussian)} & \multirow{ 2}{*}{$\textbf{28.15}$} &$\textbf{29.52}$  $\textbf{0.900}$ &  $\textbf{28.56}$  $\textbf{0.824}$ & $\textbf{25.71}$  $\textbf{0.848}$ & $28.72$  $\textbf{0.849}$ & $\textbf{26.45}$ $\textbf{0.908}$ &  $\textbf{29.78}$  $\textbf{0.825}$ & $33.08$  $0.891$ & $\textbf{32.42}$  $\textbf{0.896}$ & $\textbf{28.90}$ $0.829$ \\
\hline
\end{tabularx}}
\label{Inpaintab1}
\end{table*}

\subsection{Choice of guide image}
In Section \ref{propsolver}, the guide image $\u$ is derived from the measurements $\y$ using $5$ to $10$ PnP iterations, as done in \cite{sreehari2016plug,Teodoro2019PnPfusion,nair2021fixed,gavaskar2021plug}.
Note that the final reconstruction depends on the regularizer $\phi_{\W}$ in \eqref{regularizer} and hence $\W$, and $\W$ in turn depends on the choice of guide. Intuitively, we should expect better reconstructions if the guide image resembles the ground truth, i.e., if we use more PnP iterations to compute the guide from the measurements---this is indeed the case with the inpainting results in Table \ref{PnPiter1}. However, this is not true in general. We empirically found that the dependence of the number of PnP iterations (used to compute the guide) on the reconstruction is complicated and depends on the image and application at hand. For example, from the deblurring results in Table \ref{PnPiter1}, we see that one or two iterations seem to give the best results. On the other hand,  for the inpainting results in Table \ref{PnPiter1}, the reconstruction is seen to  improve with the number of PnP iterations. We wish to investigate this aspect more thoroughly in future work.

\section{Experimental Results}
\label{exp}

To understand the reconstruction capability of the proposed algorithm in relation to state-of-the-art methods, we apply our algorithm to three different inverse problems---inpainting, superresolution, and deblurring. In step \eqref{oursolver} of Algorithm \ref{propalgo}, we use the GCROT solver; we use just $5$ iterations which takes about $10$ seconds for a $512 \times 512$ image. All experiments were performed on a $2.3$ GHz, $36$ core machine (no GPUs are used). The input images in Fig.~\ref{Inputimages} are used for comparisons. Timings are reported to highlight the speedup obtained using our algorithm. 


\begin{table*}[!htp]
\caption{Experiment 2: Comparison of timing, PSNR and SSIM for Inpainting with $80$ percent missing pixels (noise level $0.04$ ).}
\centering
\resizebox{\textwidth}{!}{%
\begin{tabularx}{\textwidth}{ |c|| c|| *{10}{Y|} }
\hline
Method &  \scriptsize \textbf{time(s)} & \scriptsize \textit{barbara} & \scriptsize\textit{boat} & \scriptsize \textit{cameraman} & \scriptsize \textit{couple} & \scriptsize \textit{fingerprint} &  \scriptsize \textit{hill} &  \scriptsize \textit{house} & \scriptsize \textit{lena} & \scriptsize \textit{man}\\
\hline
\multirow{ 2}{*}{EPLL \cite{zoran2011learning}} & \multirow{ 2}{*}{$257.30$} & $24.76$  & $26.59$   & $23.88$   & $26.70$  & $23.31$ &  $27.86$  & $29.10$  & $28.95$  & $26.99$\\
& &  $0.720$ & $0.720$ & $0.727$ & $0.733$ & $0.812$ & $0.710$ & $0.773$ & $0.777$ & $0.724$ \\
\hline
\multirow{ 2}{*}{IRCNN \cite{zhang2017learning}} & \multirow{ 2}{*}{$31.10$} & $25.94$  & $26.86$ & $\textbf{24.75}$ & $26.98$ & $24.71$ & $27.90$ & $30.61$ & $29.94$ &  $27.01$\\
& &  $0.781$ & $0.741$ &  $0.794$ &  $0.757$ &  $0.834$ &  $0.726$ & $0.845$ &  $0.833$ &  $0.725$ \\
\hline
\multirow{ 2}{*}{IDBP-BM3D \cite{Tirer2019_iter_denoising}} & \multirow{ 2}{*}{$39.97$} & $25.03$   & $27.02$ & $24.68$ &  $\textbf{27.22}$ & $24.99$ & $28.00$& $\textbf{31.62}$  & $30.14$  & $27.24$ \\
& & $0.755$ &  $0.731$ &  $0.786$ &  $0.759$ &  $0.836$ & $0.708$ & $0.850$ & $0.835$ & $0.727$ \\
\hline
\multirow{ 2}{*}{IDBP-CNN \cite{Tirer2019_iter_denoising}} & \multirow{ 2}{*}{$63.10$} & $26.12$  & $26.95$   & $23.94$  & $27.04$  & $24.50$ &  $27.93$   & $31.16$ & $\textbf{30.17}$  & $27.21$ \\
& &  $0.753$ & $\textbf{0.756}$ & $0.791$  & $0.773$ & $0.857$ & $\textbf{0.734}$  & $\textbf{0.851}$  & $\textbf{0.849}$ & $\textbf{0.735}$ \\
\hline
\multirow{ 2}{*}{Proposed} & \multirow{ 2}{*}{$\textbf{28.30}$} & $\textbf{27.67}$ & $\textbf{27.16}$ & $24.64$ & $27.08$ & $\textbf{25.92}$  & $\textbf{28.04}$ & $30.70$  & $30.10$ & $\textbf{27.39}$\\
& & $\textbf{0.791}$ & $0.746$ & $\textbf{0.795}$ & $\textbf{0.761}$ & $\textbf{0.886}$ & $0.732$ &  $0.835$ & $0.833$ & $0.717$ \\
\hline
\end{tabularx}}
\label{Inpaintab2}
\end{table*}

\begin{figure*}
\centering
  \subfloat[Ground truth]{\includegraphics[width=0.16\linewidth]{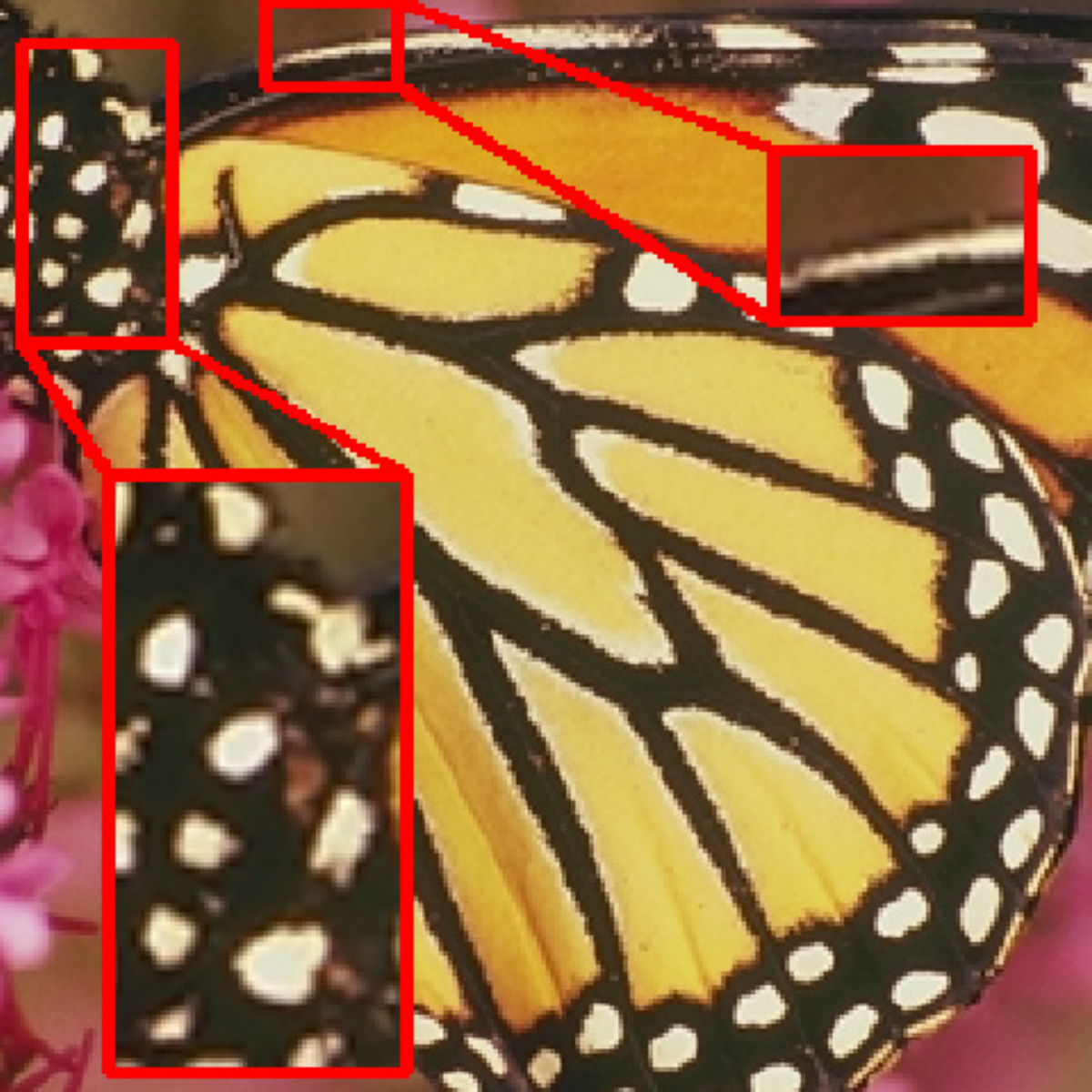}} \hspace{0.05mm}
  \subfloat[Bicubic \cite{chambolle2004algorithm}]{\includegraphics[width=0.16\linewidth]{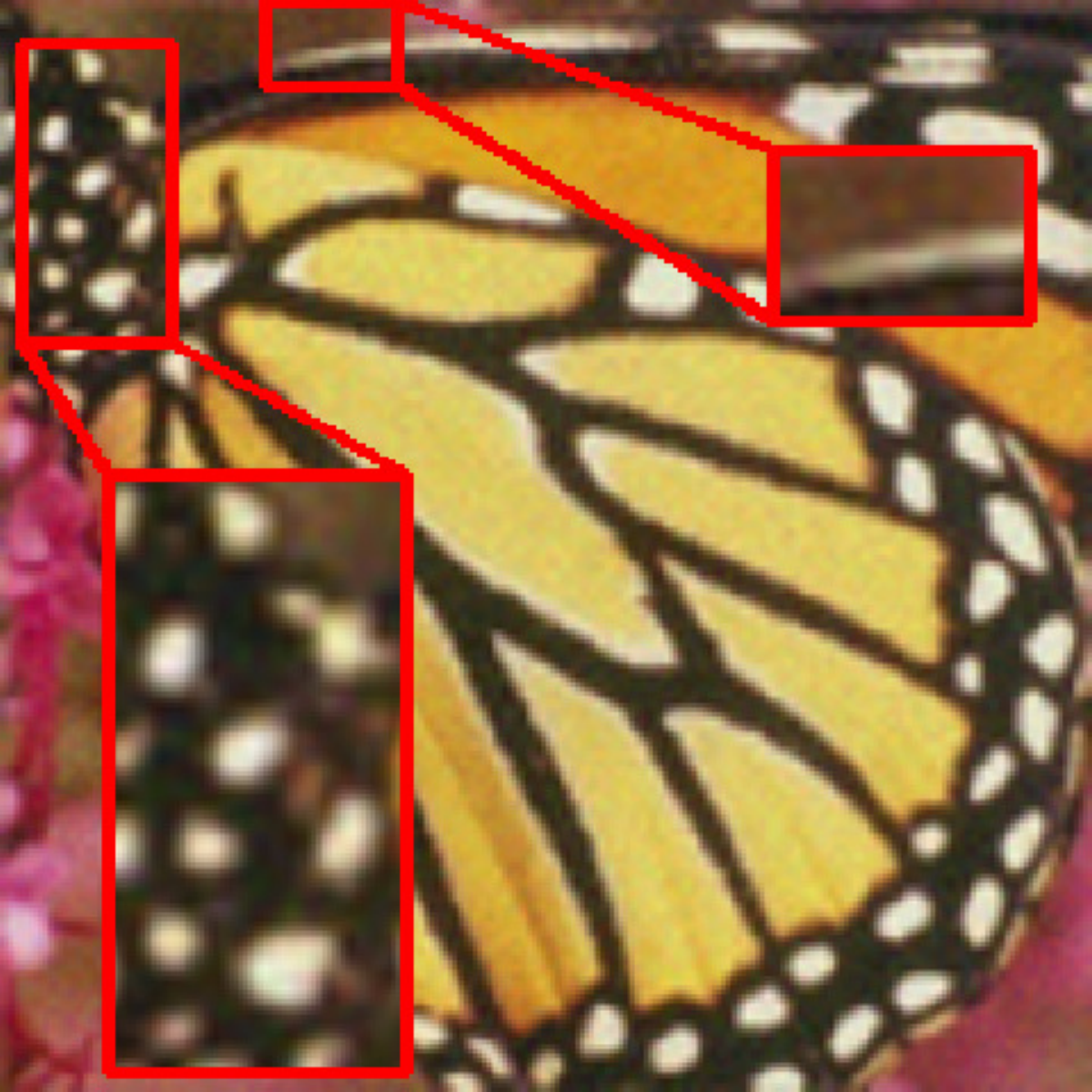}} \hspace{0.05mm}
  \subfloat[TV \cite{chambolle2004algorithm}]{\includegraphics[width=0.16\linewidth]{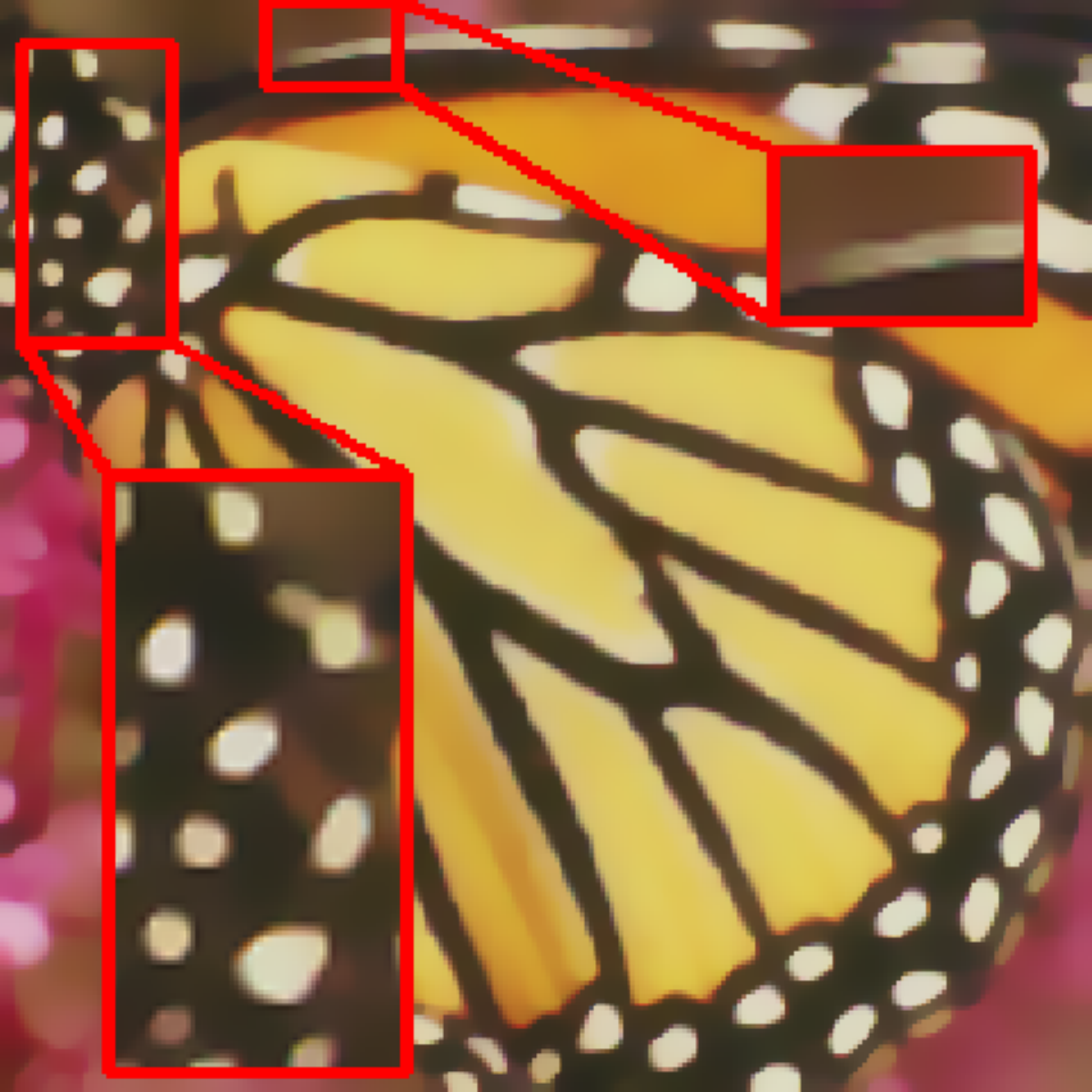}}  \hspace{0.05mm}  
  \subfloat[BM3D \cite{dabov2007image}]{\includegraphics[width=0.16\linewidth]{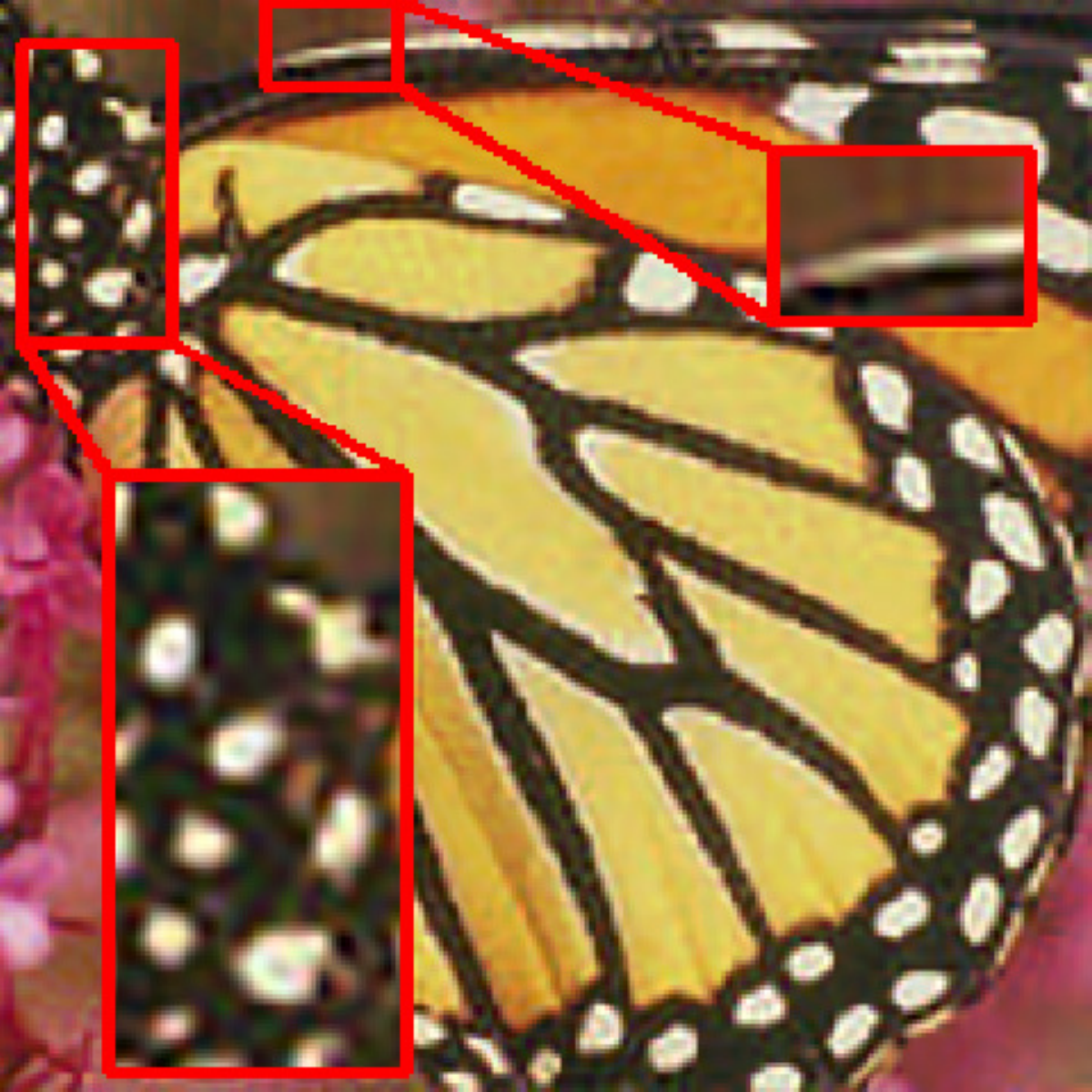}} \hspace{0.05mm}
  \subfloat[DnCNN \cite{Ryu2019_PnP_trained_conv}]{
	   \includegraphics[width=0.16\linewidth]{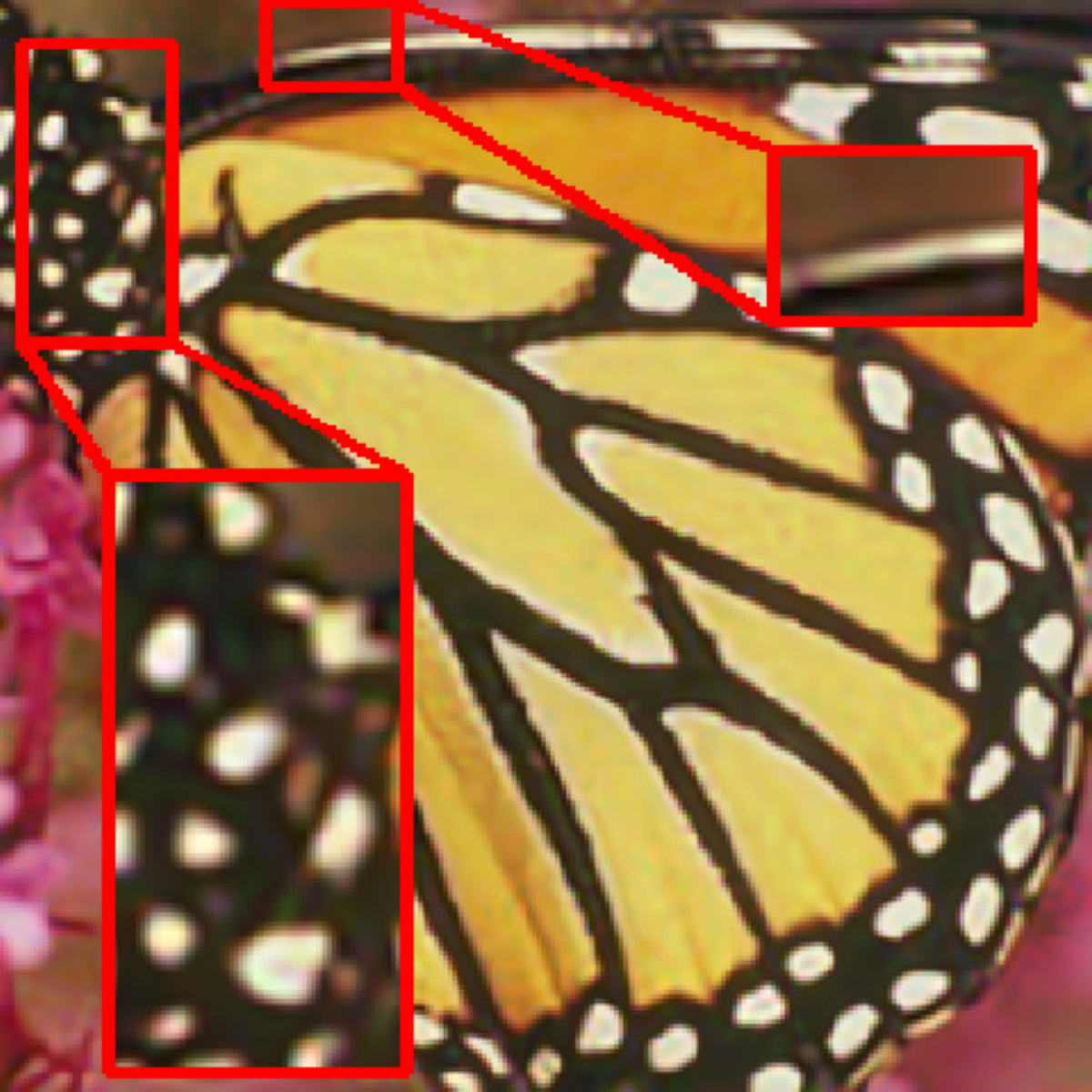}} \hspace{0.05mm}
  \subfloat[Proposed ]{\includegraphics[width=0.16\linewidth]{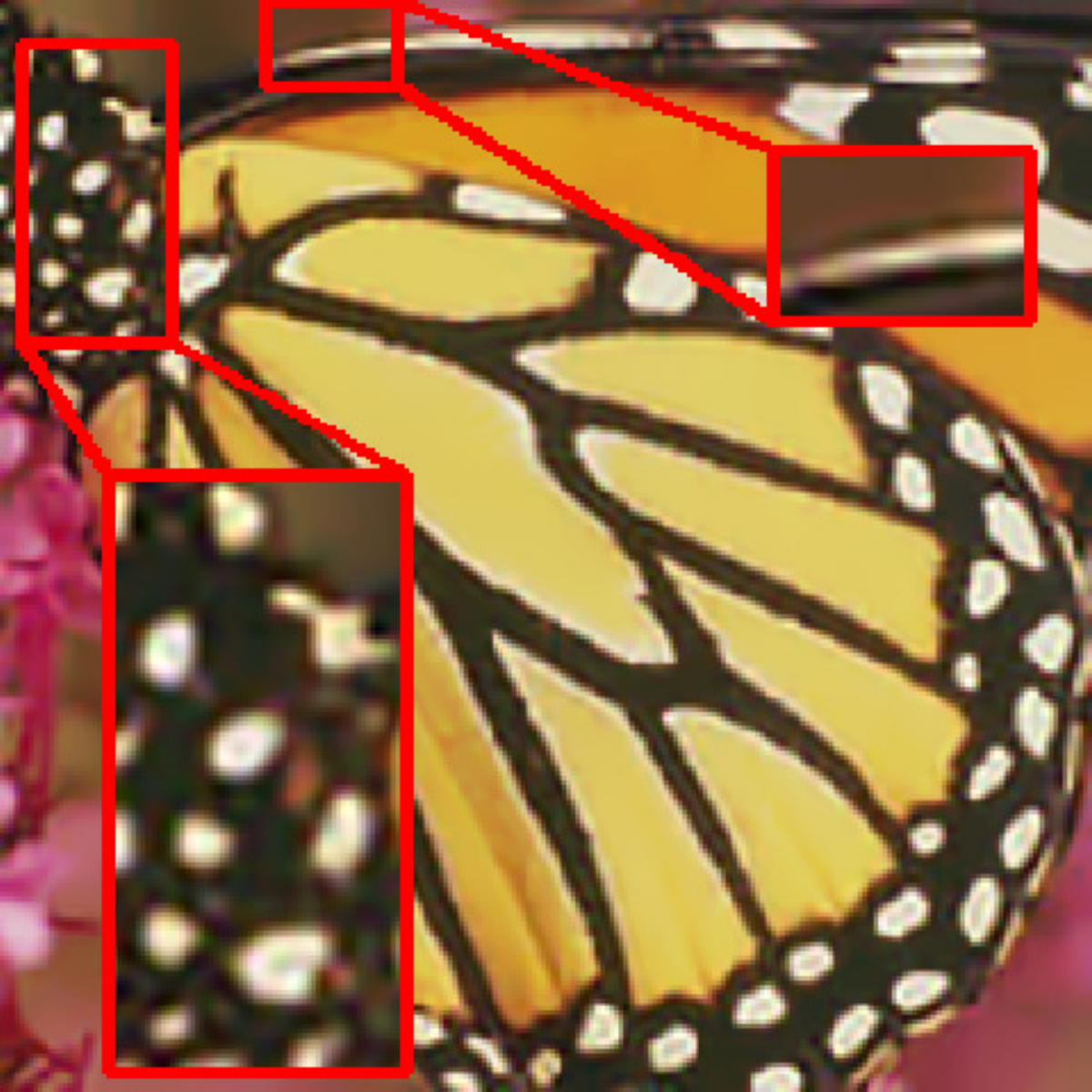}}
\caption{Results for $2\times$ image superresolution (using $9 \times 9$ Gaussian blur with standard deviation $1$ and  noise level $5/255$). Compared to BM3D (see zoomed regions), the quality is perceptibly better for PnP regularization, both for DnCNN and the proposed method ($\rho = 2$). PSNR(dB), SSIM values: (b) $22.57$, $0.78$ (c) $24.3$, $0.83$ (d) $26.2$, $0.86$ (e) $\textbf{27.4}$, $0.89$; and (f) $27.03$, $\textbf{0.90}$.}
\label{Superresfig}
\end{figure*}

\begin{table*}[!htp]
\caption{Comparison of timing and PNSR for superresolution ($9 \times 9$ Gaussian blur with standard deviation $1$, and  noise level $5/255$).}
\centering
\resizebox{\textwidth}{!}{
\begin{tabularx}{\textwidth}{ |c|| c || *{10}{Y|} }
\hline
Methods & \scriptsize \textbf{time (s)}& \scriptsize \textit{barbara} & \scriptsize \textit{boat} & \scriptsize \textit{cameraman} & \scriptsize \textit{couple} & \scriptsize \textit{fingerprint} & \scriptsize \textit{hill} & \scriptsize \textit{house} & \scriptsize \textit{lena} & \scriptsize \textit{man} \\
\hline
 \multicolumn{11}{c}{$K=2$} \\
\hline
SR\cite{yang2008image} & $298.6$ & $23.61$ & $26.25$ & $23.71$ & $26.15$ & $23.80$ & $27.41$ & $27.71$  & $28.07$ & $26.99$ \\
GPR\cite{he2011single} & $390.12$ & $23.82$ & $26.81$ & $23.91$ & $26.63$ & $24.05$ & $28.38$ & $29.16$ & $29.54$ & $27.78$  \\
NCSR\cite{dong2012nonlocally} & $256.16$ & $24.67$ & $28.40$ & $26.22$ & $28.02$ & $27.74$ & $28.50$ & $29.85$ & $30.43$ & $28.75$  \\
PnP-BM3D \cite{CWE2017} & $48.65$ & $24.64$ & $29.41$ & $26.73$ & $29.22$ & $28.82$ & $29.82$ & $32.65$ & $32.76$ &  $29.66$  \\
PnP-DnCNN \cite{Ryu2019_PnP_trained_conv} & $39.45$ & $24.50$ & $\textbf{29.57}$ & $\textbf{27.36}$ & $29.29$ & $29.06$ & $29.97$ & $32.12$ & $32.58$ & $\textbf{29.98}$ \\
Proposed & $\textbf{28.12}$ & $\textbf{24.96}$ & $\textbf{29.57}$ & $26.74$ & $\textbf{29.38}$ & $\textbf{29.28}$ & $\textbf{30.07}$ & $\textbf{32.68}$ & $\textbf{32.83}$ & $\textbf{29.98}$  \\
\hline
\multicolumn{11}{c}{$K=4$} \\
\hline
SR\cite{yang2008image} & $298.6$ & $20.67$ & $21.30$ & $18.86$ & $21.51$ & $16.37$ & $23.15$ & $22.19$  & $22.85$ & $22.26$  \\
GPR\cite{he2011single} & $390.12$ & $21.55$ & $22.68$ & $19.90$ & $22.77$ & $17.70$ & $24.57$ & $23.51$ & $24.37$ & $23.63$ \\
NCSR\cite{dong2012nonlocally} & $256.16$ & $22.86$ & $24.38$ & $22.04$ & $24.18$ & $22.31$ & $25.01$ & $26.30$ & $26.90$ & $25.41$  \\
PnP-BM3D \cite{CWE2017} & $48.65$ & $23.62$ & $\textbf{25.75}$ & $23.06$ & $\textbf{25.30}$ & $23.48$ & $27.17$ & $\textbf{29.14}$ & $29.42$ & $26.86$\\
PnP-DnCNN \cite{Ryu2019_PnP_trained_conv} & $39.45$ & $23.66$ & $25.73$ & $\textbf{23.39}$ & $25.27$ & $23.61$ & $27.19$ & $29.05$ & $29.44$ & $26.92$  \\
Proposed & $\textbf{28.12}$ & $\textbf{23.67}$ & $25.64$ & $23.05$ & $\textbf{25.30}$ & $\textbf{23.65}$ & $\textbf{27.20}$ & $28.46$ & $\textbf{29.45}$ & $\textbf{26.93}$  \\
\hline
\end{tabularx}
}
\label{Superrestab}
\end{table*}

\begin{table}[!htp]
\caption{PSNR and SSIM on BSDS300 dataset for the inpainting experiments in Tables \ref{Inpaintab1} and \ref{Inpaintab2}.}
\centering
\resizebox{0.48\textwidth}{!}{%
\begin{tabular}{ |c|| c | c| c| }
\hline
Methods &  PnP-BM3D \cite{CWE2017} & PnP-DnCNN \cite{Ryu2019_PnP_trained_conv} & Proposed\\
\hline
\multirow{ 2}{*}{Experiment 1} & $27.01$ & $27.14$ & $\textbf{27.15}$\\
&  $0.8232$ & $\textbf{0.8265}$ & $0.8242$ \\
\hline
\multirow{ 2}{*}{Experiment 2} & $25.59$  & $\textbf{26.05}$  & $26.00$\\
& $0.7485$ &  $\textbf{0.7584}$ & $0.7493$\\
\hline
\end{tabular}}
\label{Inpaintab300}
\end{table}

\subsection{Inpainting} 

The problem in inpainting is to estimate missing pixels with known locations in an image; see \cite{Tirer2019_iter_denoising} for the forward operator $\F$  in this case. A sample real-world application for removing text from an image was already shown in Fig.~\ref{inpaintreal}. Further, in Tables \ref{Inpaintab1} and \ref{Inpaintab2}, we present an extensive comparison with state-of-the-art methods for image inpainting using just 20$\%$ pixels. Note that we are competitive with CNN-based methods \cite{zhang2017learning} and \cite{Tirer2019_iter_denoising}, though we take less time. We note that solving the  linear system in step $4$ of our algorithm takes about $10$ seconds and the  remaining time is taken by step $1$ to construct the guide image. For the result in Fig.~\ref{inpaintfig}, we see that our method can preserve thin edges which the compared methods are unable to do. In Table \ref{Inpaintab300}, we shown that our algorithm is competitive with top performing methods on the BSDS300 dataset comprising of  over $300$ images \cite{MartinFTM01}.

At this point, we wish to justify our choice of NLM as the preferred kernel denoiser. In Table \ref{Inpaintab1}, we have compared the results using different kernel denoisers: bilateral, Yaroslavsky, standard NLM (Gaussian kernel) and NLM (Laplacian kernel). Since we consistently get the best results for standard NLM,  we have used this denoiser for all experiments.
\subsection{Superresolution}

A widely used model for image superresolution is $\F = \mathbf{S}\B$, where $\B \in \Re^{n \times n}$ is a blurring operator, $\mathbf{S} \in \Re^{m \times n}$ is a subsampling operator, and $m=n/K$ where $K \geqslant 1$\cite{CWE2017,ng2010solving}. In Table~\ref{Superrestab}, we compare with state-of-the art methods for $K=2$ and $4$, and using a Gaussian blur. Similar to inpainting,  we are able to compete with state-of-the-art algorithms, while being faster. In Fig.~\ref{Superresfig}, we compare our  method with PnP-ADMM by plugging different denoisers. Notice that our method and PnP regularization using DnCNN \cite{Ryu2019_PnP_trained_conv}  gives the best result. Unlike BM3D and TV, we are able to restore even fine features and do not over smooth the image. 

\subsection{Deblurring}

The forward operator $\F$ is a blur in this case. We show a deblurring application in Fig.~\ref{deblurfig} where the blur kernel is asymmetric. Our method does not  produce ringing artefacts as seen in a state-of-the-art method \cite{zhang2017learning}. In  Table~\ref{Deblurtab}, we perform three experiments as in \cite{schuler2013machine} for different Gaussian blur and noise standard deviations. The results are averaged over the images in Fig.~\ref{Inputimages}. 
\begin{enumerate}
\item[(1)] Gaussian blur $25 \times 25$ and standard deviation $1.6$, and noise $\sigma=0.04$. 
\item[(2)] Gaussian blur $25 \times 25$ and standard deviation $1.6$, and noise $\sigma=2/255$.
\item[(3)] Gaussian blur $25 \times 25$ and standard deviation $3$, and noise $\sigma=0.04$. 
\end{enumerate}

It is seen from Table~\ref{Deblurtab} that the proposed method is competitive with state-of-the-art deblurring methods \cite{dong2012nonlocally, zhang2017learning}. In Table \ref{Deblurtab300}, we show that our method is generally competitive and can sometimes outperform the top methods \cite{dong2012nonlocally,zhang2017learning} on the BSDS300 dataset. 

\subsection{Discussion}

It might seem surprising  that NLM  is able to compete with a pretrained deep denoiser for image regularization; after all, the denoising quality of DnCNN is generally a few dBs better than NLM. In this regard, we note that  the exact relation between denoising capacity and the final restoration quality (within the PnP framework) is not well understood. For example, although DnCNN is more powerful than BM3D \cite{zhang2017beyond}, it is known that plugging BM3D denoiser within a PnP algorithm can produce better results  than DnCNN \cite{Ryu2019_PnP_trained_conv,teodoro2019image}. Similarly, NLM has been shown to outperform BM3D (which is more powerful than NLM) for some applications  \cite{sreehari2016plug,sreehari2015rotationally}. 

\begin{figure*}
\centering
  \subfloat[Measurements]{\includegraphics[width=0.32\linewidth]{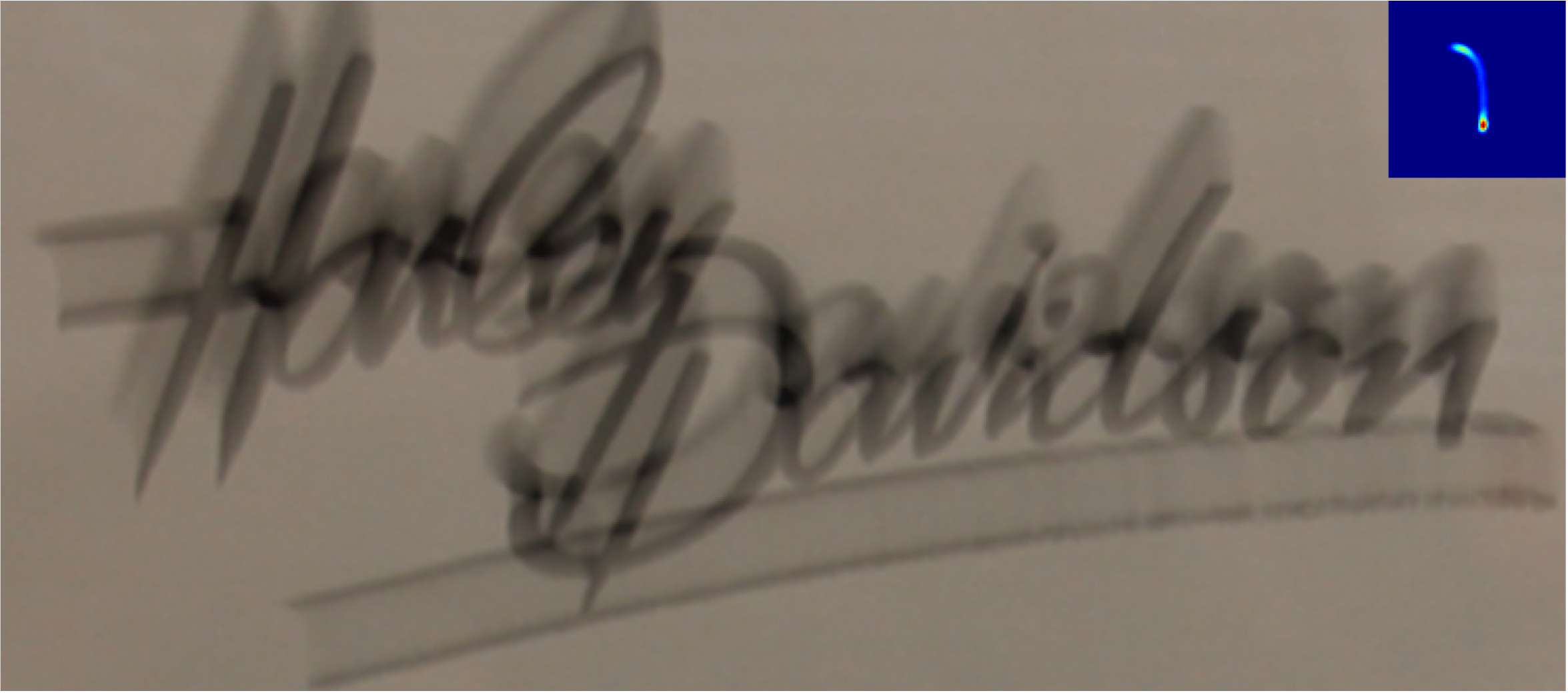}} \hspace{0.1mm}
   \subfloat[IRCNN \cite{zhang2017learning}]{\includegraphics[width=0.32\linewidth]{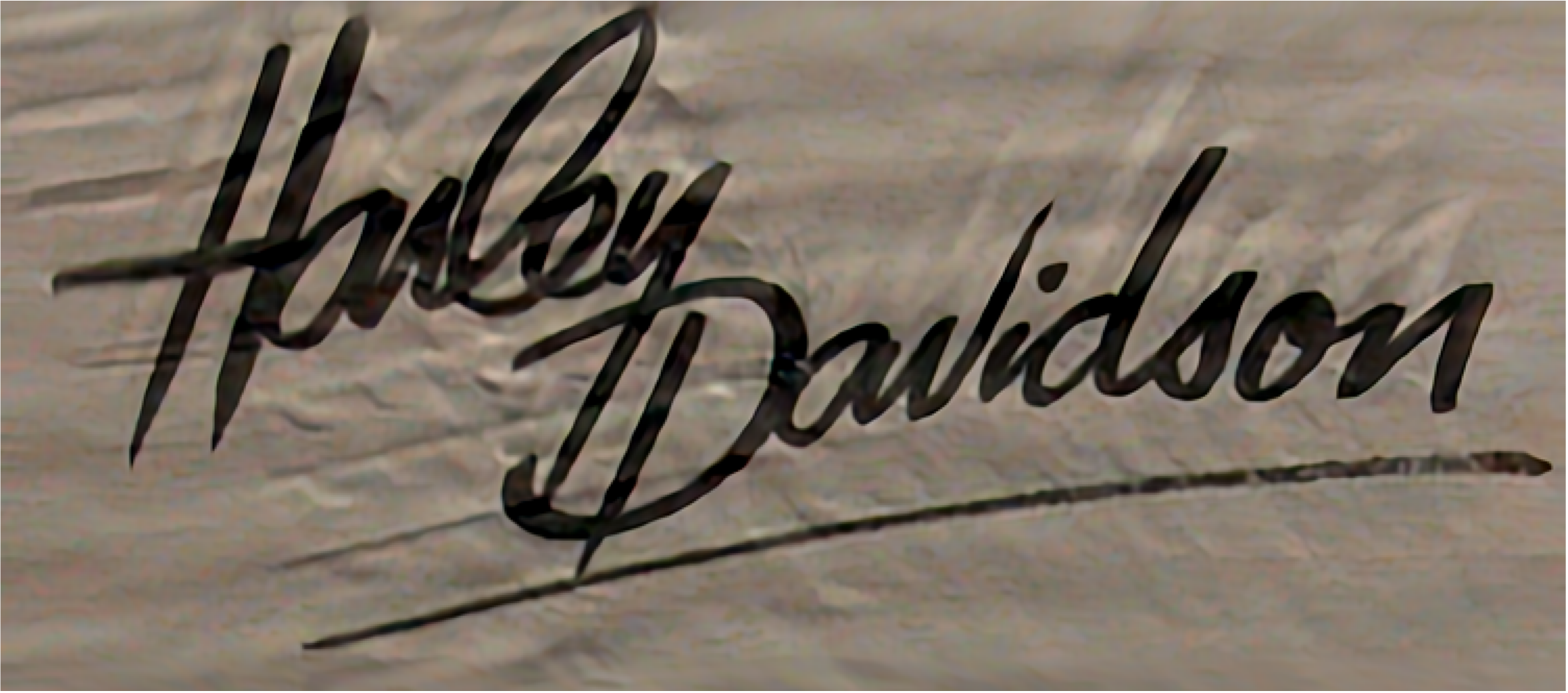}} \hspace{0.1mm} 
  \subfloat[Proposed]{\includegraphics[width=0.32\linewidth]{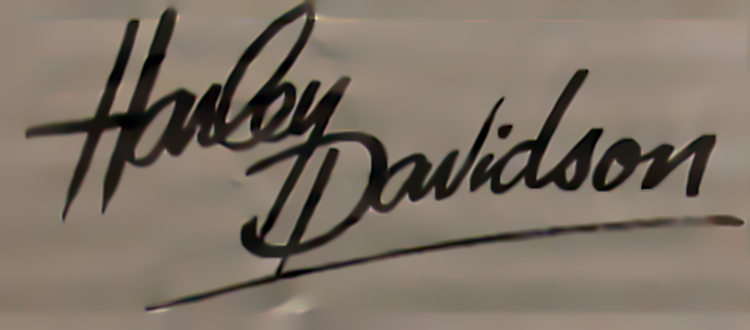}} 
\caption{Image deblurring with an asymmetric blur kernel in \cite{zhang2017learning}. Our result is visibly better than the learning-based method in \cite{zhang2017learning}.}
\label{deblurfig}
\end{figure*}
\begin{table}
\caption{Comparison of PSNR for deblurring.}
\centering
\resizebox{0.48\textwidth}{!}{
 \begin{tabular}{ | c| c| c| c|} 
 \hline
Method & Exp-I & Exp-II & Exp-III \\
 \hline
 EPLL \cite{zoran2011learning} & $24.04$ & $26.64$ & $21.36$ \\
 MLP \cite{schuler2013machine} & $24.76$ & $27.23$ & $22.20$ \\
 FlexISP \cite{heide2014flexisp} & $24.32$ & $26.84$ & $21.99$ \\
 LUT \cite{krishnan2009fast}] & $24.17$ & $26.60$ & $21.73$ \\ \cite{Meinhardt2017_learning_prox_op} & $24.51$ & $27.08$ & $21.83$ \\
 DEB-BM3D \cite{dabov2007image} & $24.19$ & $26.30$ & $21.48$ \\
 NCSR \cite{dong2012nonlocally} & $26.62$ & $30.03$ & $24.51$ \\
 IDD-BM3D \cite{danielyan2011bm3d} & $24.68$ & $27.13$ & $21.99$ \\
 IRCNN \cite{zhang2017learning} & $\textbf{27.93}$ & $\textbf{30.43}$ & $\textbf{25.67}$ \\
 Proposed  & $\textbf{27.70}$ & $\textbf{30.40}$ & $\textbf{25.62}$ \\
  \hline
 \end{tabular}
 }
 \label{Deblurtab}
 \end{table}

\begin{table}
\caption{PSNR and SSIM  on the  BSDS300 dataset for different deblurring  experiments.}
\centering
\resizebox{0.45\textwidth}{!}{%
\begin{tabular}{ |c|| c | c | c| }
\hline
Methods &  NCSR \cite{dong2012nonlocally} & IRCNN \cite{zhang2017learning}  & Proposed \\
\hline
\multirow{ 2}{*}{{Exp 1}} & $25.56$ & $\textbf{26.75}$ & $26.25$\\
& $0.6707$ & $0.7300$ & $\textbf{0.7351}$\\
\hline
\multirow{ 2}{*}{{Exp 2}} & $28.97$ & $\textbf{29.25}$ & $28.57$\\
& $0.8338$ & $0.8340$ & $\textbf{0.8371}$\\
\hline
\multirow{ 2}{*}{{Exp 3}} & $24.01$ & $\textbf{24.83}$ & $24.33$\\
& $0.5527$ & $0.6441$ & $\textbf{0.6473}$\\
\hline
\end{tabular}}
\label{Deblurtab300}
\end{table}

\section{Conclusion}
\label{conc}

This work builds on the observation that PnP regularization using kernel denoisers amounts to solving the classical regularization problem of minimizing  $f + \Phi$, where $f$ is the loss term and $\Phi$ is a convex regularizer. We showed that for linear inverse problems with quadratic $f$, the first-order optimality condition for this problem can be reduced to a linear system, which is solvable and admits a unique solution for deblurring, superresolution, and inpainting. Instead of performing PnP iterations, we proposed to directly solve this linear system. Indeed, using efficient Krylov solvers, we could solve this linear system at a superlinear rate which is a big jump from the sublinear convergence guarantee of first-order PnP algorithms. We validated the speedup in practice using deblurring, superresolution, and inpainting experiments. 
In terms of reconstruction quality, we were able to get close to deep learning methods.
A possible future work would be to apply our algorithm for hyperspectral imaging \cite{dian2021recent}, where kernel filters can play a vital role for efficient high-dimensional denoising \cite{nair2018fast,nair2019fast}.


\section{Appendix}
\label{Appendix} 
In this section, we state and prove the technical results in Section \ref{kr}. We first recall a few results from linear algebra and convex optimization. 

\subsection{Preliminaries}
 We will use $\N(\Z)$ and $\R(\Z)$ to denote  the null space and range space of a matrix $\Z$.

\begin{proposition}
\label{propprelim}
Let  $\Re^n = U_1 \oplus V_1 = U_2 \oplus V_2$, where $\oplus$ denotes orthogonal direct sum. Then $U_1 \subseteq U_2$ implies $V_2 \subseteq V_1$.
\end{proposition}
\begin{proof}
Fix $\x \in V_2$. By our assumption, $\x \perp \y$ for all $\y \in U_2$. Since $U_1 \subseteq U_2$, this mean that $\x$ belongs to the orthogonal complement of $U_1$, which is $V_1$ by assumption.
\end{proof}

\begin{proposition}
\label{prop2}
Let $ \A, \B  \in \Re^{n \times n}$ be symmetric positive semidefinite matrices. Then $\N( \A+\B) \subseteq \N(\A) \cup \N(\B)$ and  $\R(\A) \cup \R(\B) \subseteq \R(\A+\B)$. 
\end{proposition}
\begin{proof}
Let $\x \in \N(\A + \B)$. Then $\x^\top(\A + \B)\x = 0$. However, by assumption, $\x^\top \A\x \geqslant 0$ and $\x^\top \B\x \geqslant 0$ for all $\x$. Thus $\x^\top(\A + \B)\x = 0$ only if $\x^\top \A\x = 0$ and $\x^\top \B\x = 0$, which further implies that $\A\x=\B\x=\boldsymbol{0}$. Hence, we conclude that $\x \in \N(\A) \cap \N(\B)$ for all $\x \in \N(\A + \B)$. 

Since $\A, \B$ and $\A + \B$ are symmetric,  $\Re^n = \R(\A) \oplus \N(\A) = \R(\B) \oplus \N(\B) = \R(\A + \B) \oplus \N(\A + \B)$. We just proved that $\N(\A + \B) \subseteq \N(\A)$ and $\N(\A + \B) \subseteq \N(\B)$. Hence, by Proposition \ref{propprelim}, $\R(\A) \subseteq \R(\A + \B)$ and $\R(\B) \subseteq \R(\A + \B)$, so $\R(\A) \cup \R(\B) \subseteq \R(\A + \B)$. 
\end{proof}
The following is a first-order optimality condition for convex functions [$16$].  
\begin{theorem}
\label{firstoptimal}
Let $\psi: \Re^n \to \Re \cup \{\infty\}$ be convex. Moreover, assume that there exists a  differentiable function $q: \Re^n \to \Re$ such that $q$ restricted to $\mathrm{dom}(\psi) := \{\x \in \Re^n: \psi(\x) < \infty\}$ equals $\psi$. Then $\x^\ast \in \Re^n$ is a global minimizer of $\psi$ if and only if 
\begin{align*}
\forall \x \in \mathrm{dom}(\psi):  \quad \nabla q(\x^\ast)^\top \! (\x - \x^\ast) \geqslant 0. 
\end{align*}
\end{theorem}

\subsection{Proof of Proposition \ref{prop1}}
Recall that $\Phi_{\W}(\x) = (1/2) \x^\top \D (\I - \W)\W^{\dagger} \x$ for all $\x \in\R(\W)$. Since $\R(\W)$ is nonempty, $\Phi_{\W}$ is proper. 

From the definitions of $\W$ and $\W^{\dagger}$, we can easily check that $\D(\I- \W)\W^{\dagger} = \P \boldsymbol{\Sigma} \P^\top\!$, where $\P:=\D^{1/2} \Q$ and $\boldsymbol{\Sigma}$ is diagonal where $\boldsymbol{\Sigma}_{ii} := (1 / \BLambda_{ii}) - 1$ if $\BLambda_{ii} > 0$, and $=0$ otherwise. Clearly, $\D(\I- \W)\W^{\dagger}$ is symmetric. Moreover, it can be shown that the eigenvalues of $\W$ are always in $[0, 1]$; e.g., see [$22$, Proposition $3.12$]. This means that $\boldsymbol{\Sigma}_{ii} \geqslant 0$ for all $i$, and hence $\D(\I- \W)\W^{\dagger}$ is positive semidefinite, which further implies that $\Phi_{\W}$ is nonnegative and convex. 

Since $\Phi_\W$ is  continuous on its domain $\R(\W)$ and $\R(\W)$ is a closed set in $\Re^n$, $\Phi_\W$ is closed [$16$, Lemma $1.24$].  

\subsection{Proof of Theorem \ref{proximalmap}}

Following definition \eqref{prox}, we need to show that for any $\x \in \Re^n$, $\W\x $ is a minimizer of the function
\begin{equation*}
\z \mapsto \psi(\z) =  \frac{1}{2} \lVert \z - \x \rVert_{\D}^2 +  \Phi_{\W} (\z).
\end{equation*}
We can write this as
\begin{equation*}
\psi(\z) =  
 \begin{cases}
   q(\z), & \text{if } \z \in \R(\W), \\
    \infty, & \text{otherwise},   
\end{cases}
\end{equation*}
where
\begin{equation*}
q(\z)= \frac{1}{2} {(\z - \x)}^\top \! \D (\z - \x) + \frac{1}{2} \z^\top\! \D (\I - \W)\W^{\dagger} \z.
\end{equation*}
It is clear that $\psi$ is convex. On the other hand, $q$ is differentiable and 
\begin{equation*}
\nabla q(\z) = \D (\z - \x) + \D (\I - \W)\W^{\dagger} \z.
\end{equation*}
Since the domain of $\psi$ is $ \R(\W)$, to use Theorem \ref{firstoptimal}, it suffices to show that $ \nabla q(\W\x)^\top\! (\z - \W\x) \geqslant 0$ for all 
$ \z \in \R(\W)$, that is,
\begin{equation}
\label{e1}
\forall \u \in \Re^n: \quad \nabla q(\W\x)^\top\! (\W\u - \W\x) \geqslant 0.
\end{equation}
However, after some calcuations, we obtain
\begin{align}
\label{eq}
& \ \nabla q(\W\x)^\top\! (\W\u - \W\x) \nonumber \\ 
&= (\u-\x)^\top \W^\top\! \D \Big[(\W-\I) +  (\I - \W)\W^{\dagger} \W \Big] \x. 
\end{align}
Now, since $\W= \W \W^\dagger \W$ and  $\W^\top \D =  \D \W$,we have
\begin{equation*}
 \W^\top\! \D \Big[(\W-\I) +  (\I - \W)\W^{\dagger} \W \Big]  = \mathbf{0}.
\end{equation*}
Thus, the expression in \eqref{eq} is identically zero. This establishes \eqref{e1} and completes the proof. 


\subsection{Proof of Proposition \ref{minimizers}}
Note that $\x^\ast \in \R(\W)$ is a minimizer of the objective in \eqref{pnpopt} if and only if  $\x^\ast = \W \z^\ast$, where $\z^\ast$ is a minimizer of 
\begin{align*}
\theta(\z) := \frac{1}{2} \| \y - \F\W\z \|^2_2 + \frac{\rho}{2} \z^\top \W^\top \D (\I - \W) \z.
\end{align*}
Since $\theta$ is convex and differentiable in $\z$, it has a minimizer if and only if there exists $\z^\ast \in \Re^n$ such that $\nabla  \theta(\z^\ast)=\boldsymbol{0}$. However, we can write $\nabla  \theta(\z)=\A\z-\b$, where $\b = \W^\top\F^\top\y$ and $\A$ is of the form $\A = \Z+\Y$ where
\begin{equation*}
\Z=\W^\top\F^\top\F\W \ \ \text{ and } \ \  \Y=\rho\W^\top \D (\I - \W).
\end{equation*}
Thus, $\theta$ has a minimizer if and only if the equation $\A\z=\b$ is solvable, i.e.,  $\b \in \R(\A)$. We next show that this is indeed the case. 

By definition, $\b$ is in the range of $\W^\top\F^\top$. However, note that $ \R(\W^\top\F^\top) = \R(\Z)$, so it suffices to show that $\R(\Z) \subseteq \R(\A)$. But this follows from Proposition \ref{prop2} since $\Z$ and $\Y$ are symmetric positive semidefinite and $\A=\Y+\Z$. This completes the proof.

\subsection{Proof of Proposition \ref{uniqueness}}
Let $\Y$ and $\Z$ be as in the proof of Proposition \ref{minimizers} and $\A=\Y+\Z$. Since $\Z$ and $\Y$ are symmetric positive semidefinite, by Proposition \ref{prop2}, we can conclude that $\A$ is invertible if $\N(\Z) \cap \N(\Y) = \{\boldsymbol{0}\}$. 
Now, it can be shown that $\W$ is irreducible, nonnegative and row stochastic [$36$, Chapter $8$]. Hence, by the Perron-Frobenius theorem, $\e$ is the Perron vector and $\W\e=\e$. Also $\W$ is nonsingular by hypothesis. Hence $\N(\I - \W)$ consists of scalar multiples of the Perron vector $\e$. Since $\W$ is invertible by hypothesis and $\D$ is invertible by construction, we have $\N(\Y)=\N(\rho \W^\top \D(\I - \W))= \{t \e: t \in \Re\}$. Since $\W\e = \e$ and $\F\e \neq \boldsymbol{0}$ by hypothesis, $\F\W\e = \F\e \neq \boldsymbol{0}$ and $\N(\F\W) = \N(\W^\top \F^\top \F \W)=\N(\Z)$.  Hence $\N(\Z) \cap \N(\Y) = \{\boldsymbol{0}\}$. 

\bibliographystyle{IEEEtran}
\bibliography{refs}

\end{document}